\documentclass[letterpaper, twocolumn, accepted=2023-04-26]{quantumarticle}
\pdfoutput=1
\usepackage[utf8]{inputenc}
\usepackage[english]{babel}

\usepackage{amsmath, amssymb, amsthm}
\usepackage{physics}

\usepackage{thmtools}
\usepackage{thm-restate}

\usepackage[export]{adjustbox}
\usepackage{graphicx}
\usepackage{here}
\usepackage{color}

\usepackage{enumerate}

\usepackage{makecell}

\usepackage{kbordermatrix}
\renewcommand{\kbldelim}{(}
\renewcommand{\kbrdelim}{)}

\usepackage{hyperref}
\usepackage[numbers]{natbib}

\declaretheorem[name=Lemma]{lemma}

\newtheorem{definition}{Definition}
\newtheorem{corollary}{Corollary}
\newtheorem{theorem}{Theorem}

\begin{document}

\title{Classification of measurement-based quantum wire in\newline stabilizer PEPS}

\author{Paul Herringer} 
\affiliation{Department of Physics and Astronomy, University of British Columbia, Vancouver, Canada}
\affiliation{Stewart Blusson Quantum Matter Institute, University of British Columbia, Vancouver, Canada}
\author{Robert Raussendorf}
\affiliation{Department of Physics and Astronomy, University of British Columbia, Vancouver, Canada}
\affiliation{Stewart Blusson Quantum Matter Institute, University of British Columbia, Vancouver, Canada}

\begin{abstract}
    We consider a class of translation-invariant 2D tensor network states with a stabilizer symmetry, which we call stabilizer PEPS. The cluster state, GHZ state, and states in the toric code belong to this class. We investigate the transmission capacity of stabilizer PEPS for measurement-based quantum wire, and arrive at a complete classification of transmission behaviors. The transmission behaviors fall into 13 classes, one of which corresponds to Clifford quantum cellular automata. In addition, we identify 12 other classes. 
\end{abstract}

\maketitle

\section{Introduction}

How to harness multi-particle entanglement is a central question in quantum information processing. Among the many known protocols are quantum steering~\cite{Steer1,Steer2}, distribution of localizable entanglement \cite{LocE} in quantum networks, Bell-state distillation in quantum communication~\cite{Bennett,Briegel}, and measurement-based quantum computation (MBQC)~\cite{Raussendorf2001}.

In some of these protocols, symmetry plays an important role. For example, we observe this feature in measurement-based quantum information processing, including measurement-based quantum state transmission (also known as quantum wire)~\cite{LRE,SPT2}, quantum computation on a resource state~\cite{Raussendorf2001,AKLT1,AKLT2}, and computational phases of quantum matter~\cite{SPT1,SPT2,SPT3,SPT4,SPT6,SPT7, Bridge, Raussendorf2019,Stephen2019}.

\begin{figure}[t]
    \begin{center}
    \includegraphics[width=0.3\textwidth]{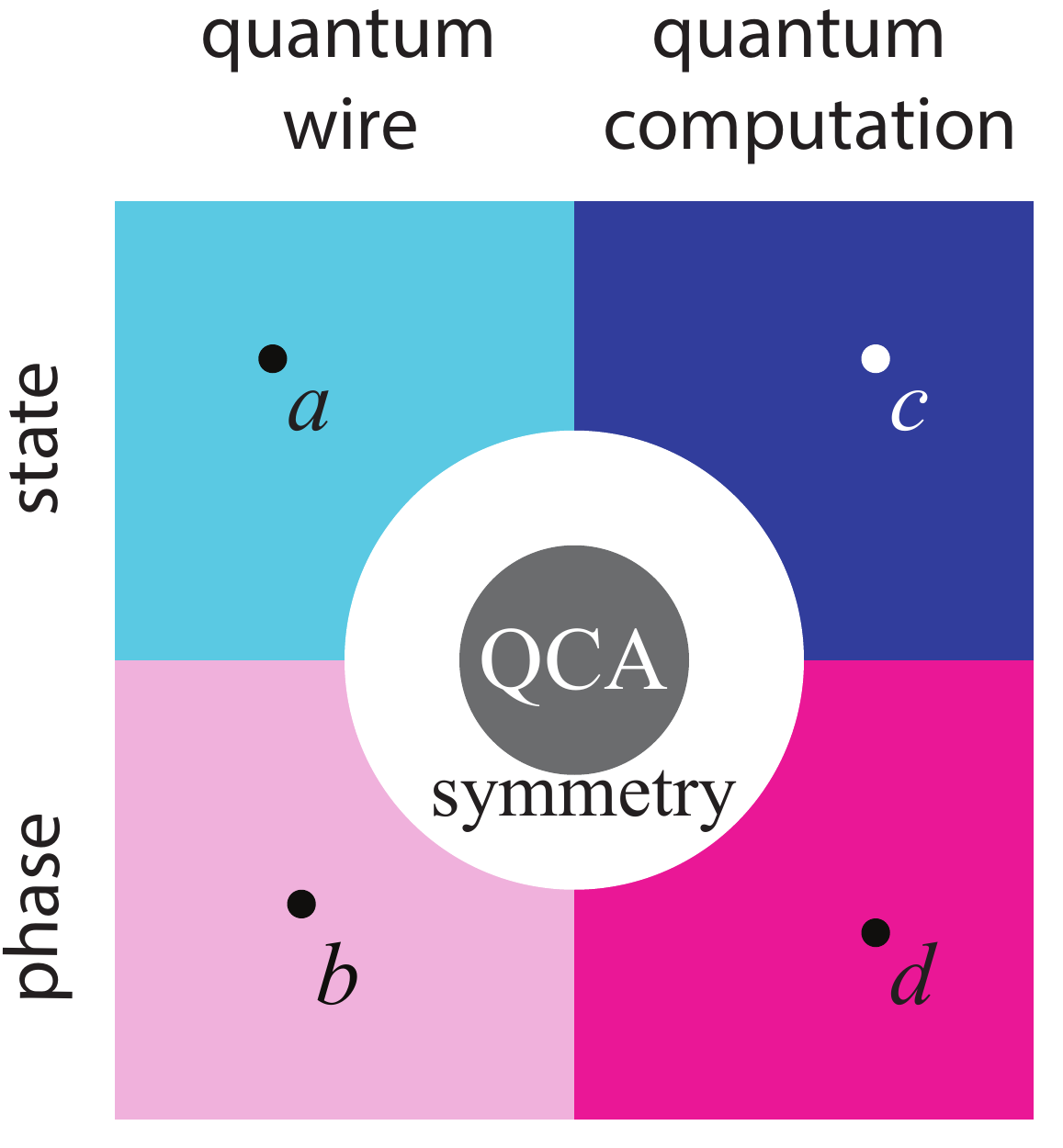}
    \caption{\label{fig:overview}Phenomenology of measurement-induced state transformations---state vs. phase and transmission vs. computation. Top left: In cluster states, arbitrarily distant qubits can be projected into Bell pairs by local measurement on other qubits, $a$=\cite{LRE}. Bottom left: the same feature is observed in an entire SPT phase surrounding the 1D cluster state, $b$=\cite{SPT2}. Top right: In spatial dimension two and higher, cluster states  give rise to universal measurement-based quantum computation, $c$=\cite{Raussendorf2001}. Bottom right: Universal measurement based quantum computation persists in entire SPT phases of matter, $d$=\cite{SPT6,SPT7, Raussendorf2019,Stephen2019}. Centre: All the phenomena displayed rely on symmetries of the PEPS tensor representing the resource state. Some, but not all, of these symmetries give rise to Clifford cellular automata. }
    \end{center}
\end{figure}

The present work addresses the role of symmetry in protocols of measurement-based state transformation. Specifically, we study a class of highly symmetric 2D tensor network states and classify their transmission capacity for quantum wire. This is a first step toward a classification of schemes of measurement-based quantum computation in higher spatial dimension, based on symmetry.

To put our work into perspective, below we summarize several measurement-based, symmetry-aided protocols and describe our classification result in more detail. See Fig. \ref{fig:overview} for a summary of the phenomenology of interest. The simplest among the phenomena displayed is the distribution of localizable entanglement, as can be observed in cluster states, for example. Therein, any pair of qubits can be brought to a Bell-entangled state by locally measuring other qubits in the state \cite{LRE}. It was later found that the creation of Bell-type entanglement, or equivalently quantum wire, works for an entire symmetry-protected topological (SPT) phase surrounding the cluster state~\cite{SPT2}. The quantum wire is a direct consequence of the symmetry representation in those phases.

Ramping up the complexity of the quantum processing task, it is also known that certain special quantum states, such as cluster states and AKLT states on various two-dimensional lattices, are universal resources for quantum computation by local measurement \cite{LRE}, \cite{AKLT1,AKLT2,AKLT3}.

The most complex notion in this area is that of computational phases of quantum matter~\cite{SPT1,SPT2,SPT3,SPT4,SPT6,SPT7, Bridge, Raussendorf2019,Stephen2019}, which extend the computational power of individual resource states to entire SPT phases surrounding them. For suitable SPT phases in dimensions 1 and 2, the computational power of MBQC resource states is uniform across the phase, i.e., all states in a given SPT phase have the same computational power. Indeed, computational phases of matter exist that support universal quantum computation~\cite{Raussendorf2019,Stephen2019,SPT6,SPT7}.

The current understanding of computational phases relies upon symmetric tensor networks and Clifford quantum cellular automata (QCA)~\cite{Raussendorf2019, Stephen2019, Schlingemann2008}. In this paradigm, the symmetries of a tensor network are associated with a Clifford QCA that persists throughout an entire SPT phase, and it is the presence of this QCA that enables MBQC with any state in the phase. This is the content of Theorem 2 in~\cite{Stephen2019}.

Our motivation for the present work is the observation that many symmetric tensor network states relevant to quantum information processing do not admit a description in terms of QCA. These include the GHZ state and the toric code \cite{Kitaev2003}, as described in Section \ref{subsec:egs}. And so, we ask: What else is there? 

To get a handle on this question, in this paper we investigate measurement-based quantum wire implemented on symmetric 2D tensor network states. In this context, we arrive at a complete classification of transmission behaviors. One of the classes identified corresponds to Clifford QCA, which have their own (finer) classification~\cite{Schlingemann2008,Schumacher2004}. In addition to Clifford QCAs, we describe 12 new classes.

The paper is organized as follows. In Sections \ref{subsec:setting}-\ref{subsec:egs}, we describe the setting for our investigation of stabilizer PEPS, detail the applications of stabilizer PEPS in measurement-based quantum information processing, and provide some examples. In Section \ref{subsec:results}, we state our main result. Section \ref{sec:classification} is devoted to proving this result, and Section \ref{sec:concl} concludes.

\section{Phenomenology} \label{sec:phenom}

\subsection{Setting} \label{subsec:setting}

Our setting is a 2D projected entangled pair state (PEPS) on a square lattice with cylindrical geometry~(Fig.~\ref{fig:setting}). The circumference and depth of the cylinder are denoted $n$ and $d$, respectively, and the state is translation invariant in the sense that the local PEPS tensor is the same for every lattice site.

\begin{figure}
  \centering
  \includegraphics[scale=0.7]{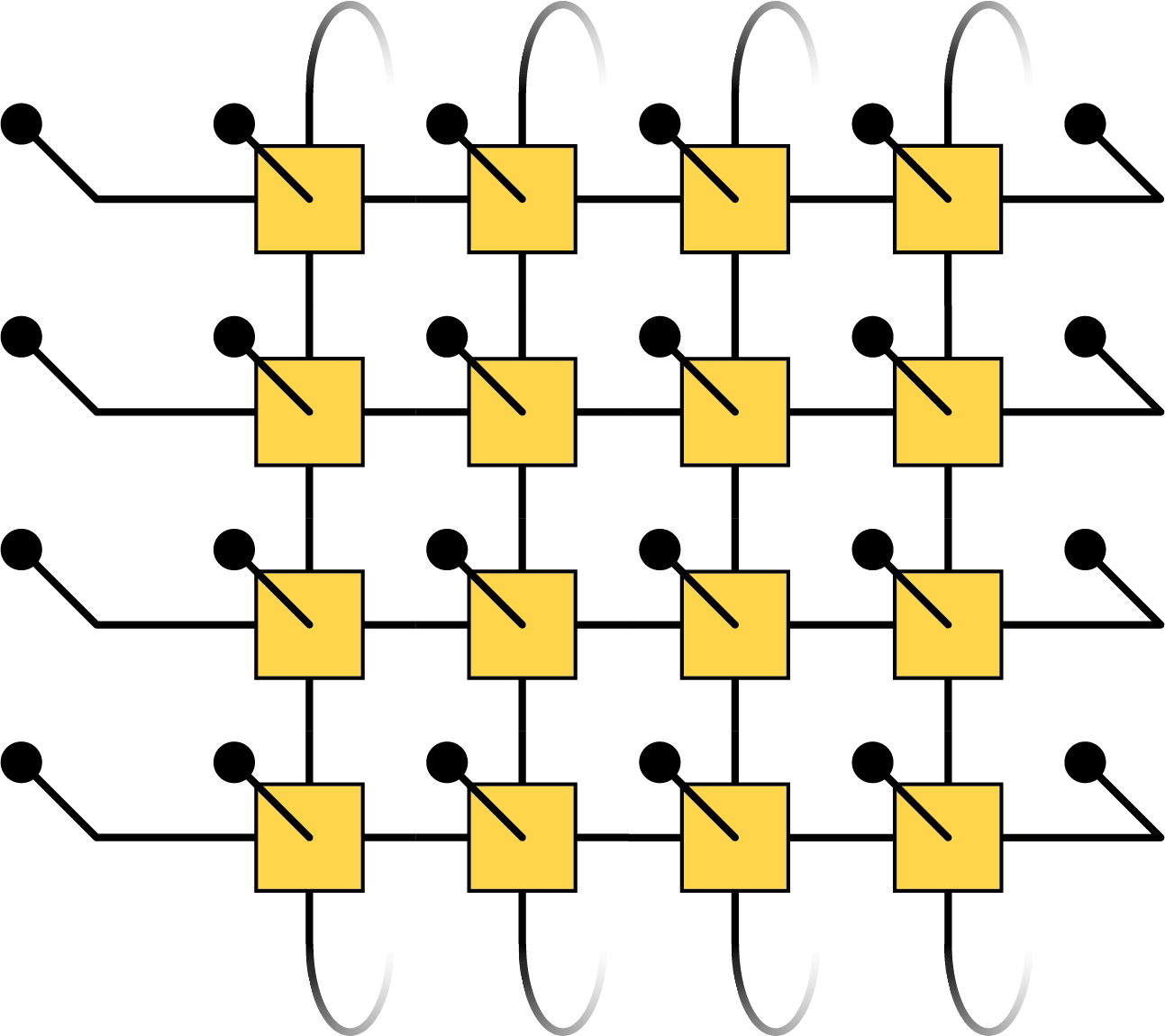}
  \caption{Cylindrical PEPS.}
  \label{fig:setting}
\end{figure}

We study a class of states characterized by local tensors which obey a set of symmetry constraints of the form
\begin{equation} \label{eqf:gen symm}
  \includegraphics[valign=c, scale=0.75]{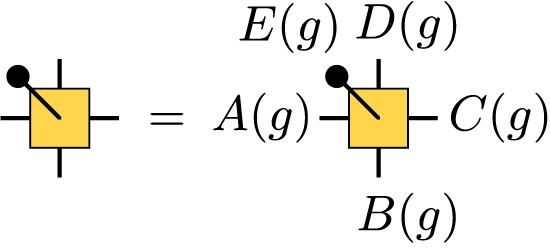}
\end{equation}
where $g$ is an element of the symmetry group of the tensor and $A(g)$ through $E(g)$ are Pauli operators that form a linear representation of $g$. In other words, the symmetry group of the local tensor is also a stabilizer group. Therefore, we call these states \textit{stabilizer PEPS}, and the local tensors \textit{stabilizer tensors}. We emphasize that stabilizer PEPS are not only theoretical but can be prepared in the lab: a stabilizer tensor is equivalent to a 5-qubit stabilizer state, and two tensors can be contracted by Bell measurements~(Fig.~\ref{fig:qubits}).

\begin{figure}
  \centering
  \begin{tabular}{c}
      \includegraphics[scale=0.6]{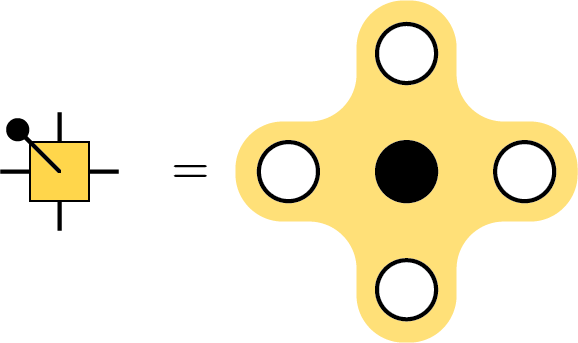} \\ 
      \rule[-6mm]{0pt}{10mm} (a) \\
      \includegraphics[scale=0.6]{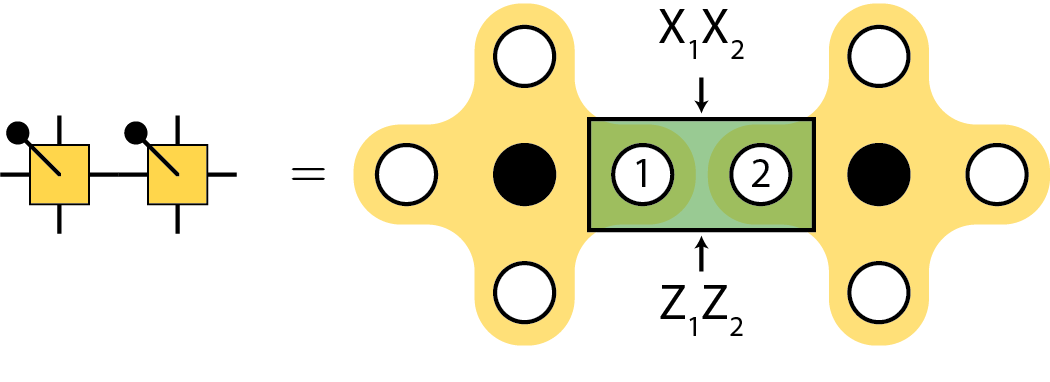} \\ 
      \rule{0pt}{4mm} (b) 
  \end{tabular}
  \caption{Physical realization of stabilizer PEPS. (a) A stabilizer tensor is a 5-qubit stabilizer state. (b) Contraction of two ``tensors'' by Bell measurement.}
  \label{fig:qubits}
\end{figure}

\subsection{Applications} \label{subsec:applications}
There are a variety of quantum information processing tasks that can be accomplished by invoking stabilizer PEPS and local measurements. In order of increasing complexity, these are (i) quantum wire or creation of long-range Bell-type entanglement, (ii) quantum computation on a resource state, and (iii) quantum computation in a suitable SPT phase. 

\paragraph{Quantum wire}
The quantum wire protocol consists of performing an $X$ measurement at every physical site in the bulk of an $n\times d$ stabilizer PEPS, which creates Bell-type entanglement between the left and right edges\footnote{The choice of $X$ as the measurement basis is arbitrary; any measurement basis can be accommodated by changing the action of the local tensor symmetries on the physical leg.}. The number of qubits worth of entanglement created in this way is a non-negative integer $C(n,d)$ such that $0\leq C(n,d)\leq n$. We call this number the transmission capacity of the stabilizer PEPS, because the creation of long-range Bell-type entanglement is equivalent to the transmission of quantum information: given long-range entanglement, quantum information may be transmitted by the standard teleportation protocol, and given the ability to transmit quantum information, a Bell pair may be created. 

\paragraph{MBQC}
If the local measurement basis is allowed to vary from site to site, it becomes possible to do more than simply transmit quantum information from one end of the cylinder to the other. Indeed, we can perform unitary operations on the $n$-qubit Hilbert space spanned by the sites around the circumference of the cylinder. For some resource states, such as the cluster state, this leads to universal quantum computation. 

\paragraph{Computational phases}
Certain stabilizer PEPS also possess subsystem symmetries, which act on subsets of the lattice rather than at every site~\cite{Stephen2019}. The family of states that can be reached from these stabilizer PEPS by local, finite-depth quantum circuits that respect the subsystem symmetries forms an SPT phase. Some of these SPT phases are known to have uniform computational power~\cite{Raussendorf2019,Stephen2019}. In other words, given an appropriate measurement protocol every state in the phase is equally useful for MBQC.

As a first step towards understanding the phenomenology of quantum information processing with stabilizer PEPS, the present paper deals with the simplest task, quantum wire.

\subsection{Examples} \label{subsec:egs}

In this section, we present three examples of well-known states that are also stabilizer PEPS. We detail their computational capability with respect to the three tasks described above, and highlight the role of symmetry and QCA.

\subsubsection{Cluster state}

The 2D cluster state is central to the study of MBQC because it is useful for quantum information processing at all three levels of complexity described in the previous section. At the simplest level, the cluster state has maximal transmission capacity for quantum wire, $C(n,d)=n$ for any depth. At the next level, MBQC, we can perform universal quantum computation on the cluster state via local measurements. Finally, the cluster state lies in an SPT phase called the cluster phase, and every state in the cluster phase is a universal resource for MBQC~\cite{Raussendorf2019}. 

The cluster state is also a stabilizer PEPS where the local tensor has the following symmetries:
\begin{equation} \label{eqf:cluster_symms_x}
  \includegraphics[valign=c, scale=0.7]{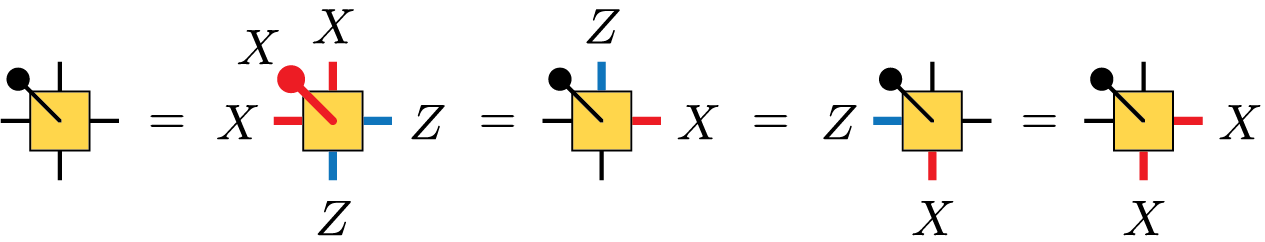}
\end{equation}
\begin{equation} \label{eqf:cluster_symms_z}
  \includegraphics[valign=c, scale=0.7]{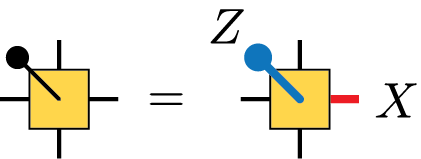}
\end{equation}
To see this explicitly, we can construct the stabilizer generators of the cluster state from its tensor symmetries (Fig. \ref{fig:stabs}a).

\begin{figure}
    \centering
    \begin{tabular}{cc}
        \includegraphics[scale=0.6]{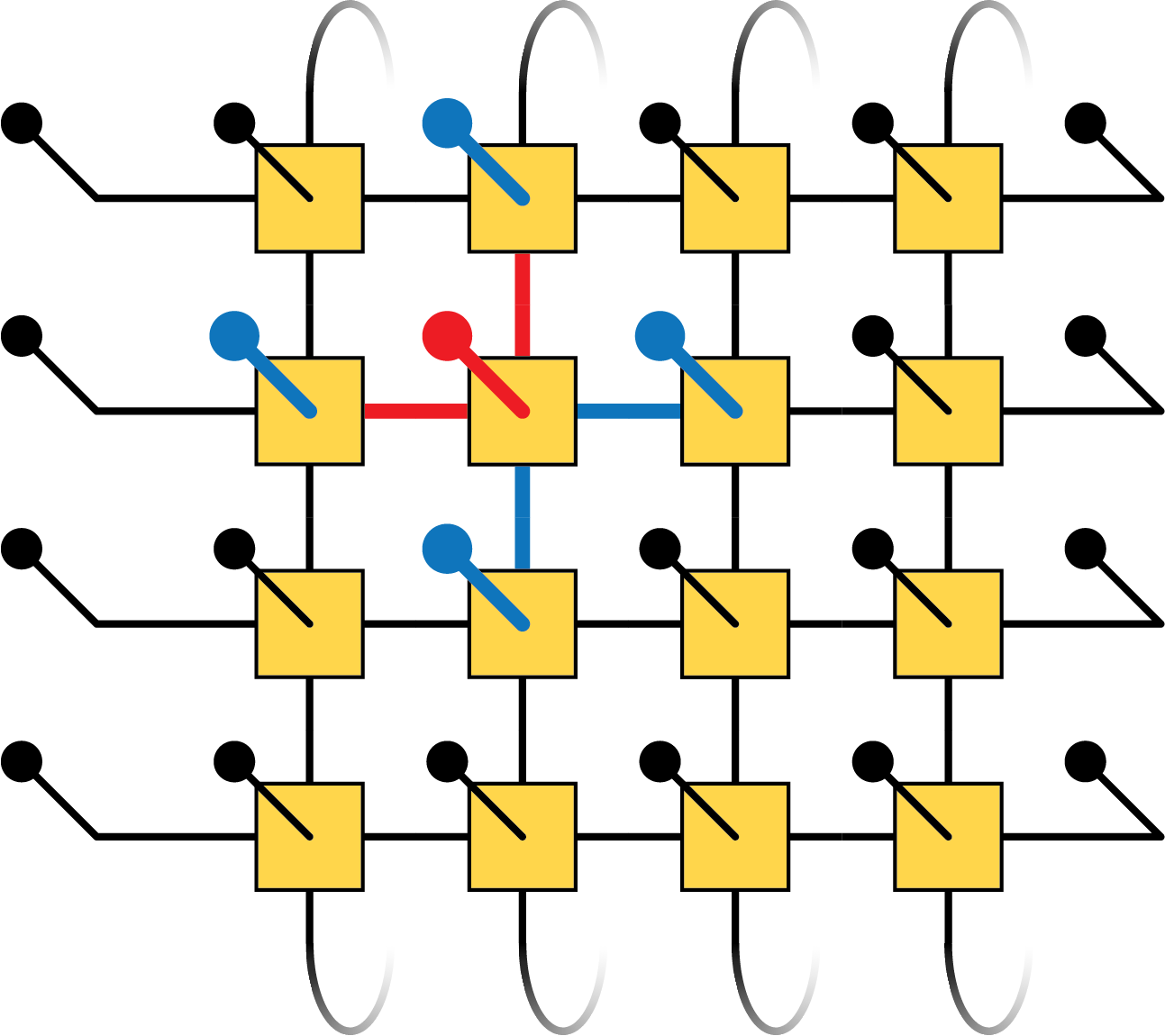} \hspace{5mm} & \includegraphics[scale=0.6]{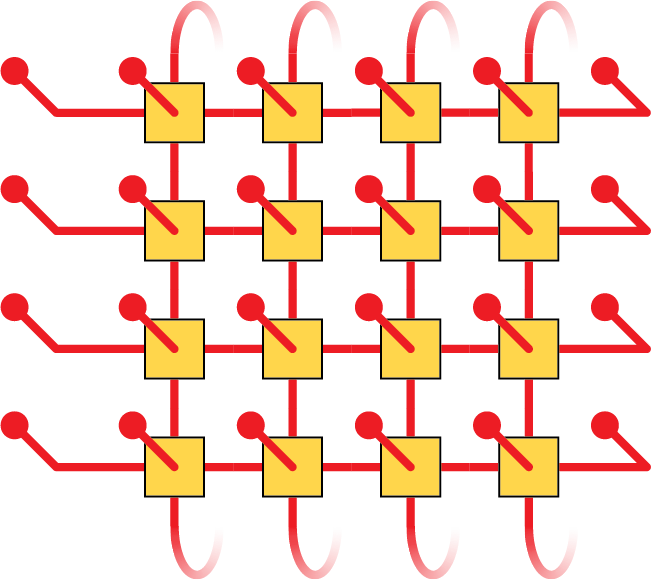} \\
        \rule[-5mm]{0pt}{9mm} (a) & (b) \\
        \includegraphics[scale=0.6]{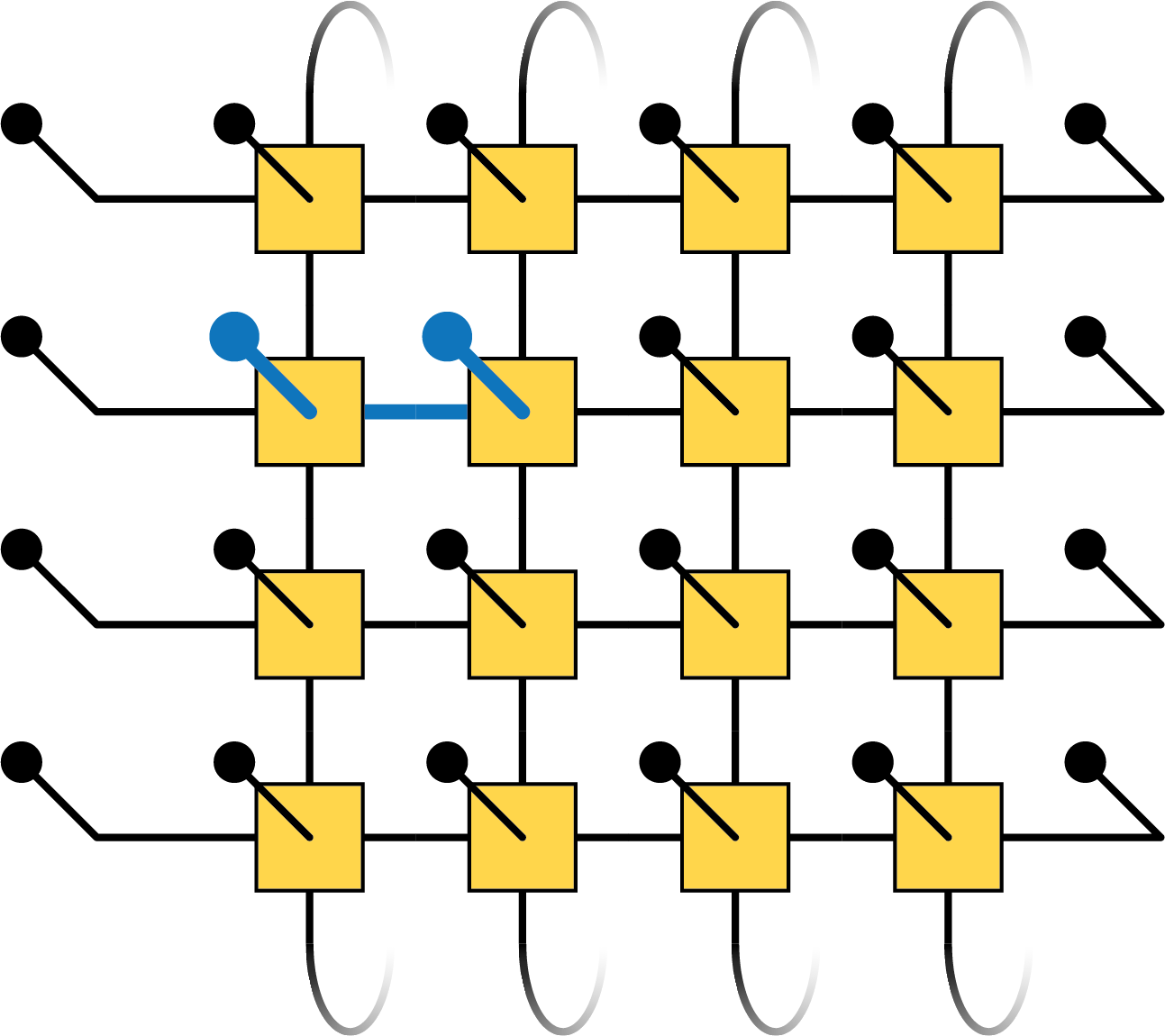} \hspace{5mm} & \includegraphics[scale=0.6]{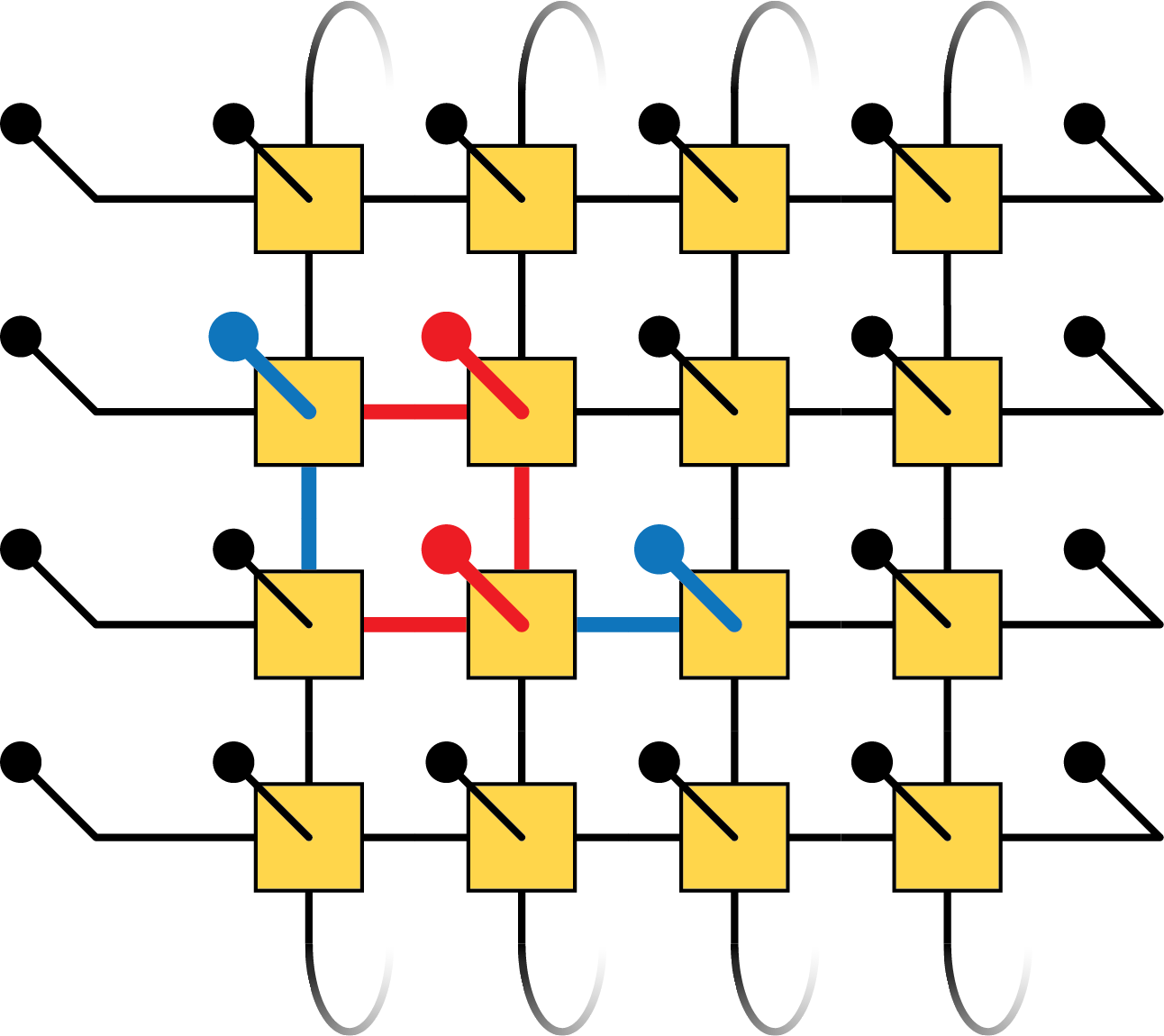} \\
        \rule{0pt}{4mm} (c) & (d)
    \end{tabular}
    \caption{Stabilizer generators for (a) the cluster state, (b, c) the GHZ state, and (d) the toric code state, all constructed from the corresponding local tensor symmetries.}
    \label{fig:stabs}
\end{figure}

The symmetry representation of the cluster state offers an illuminating perspective on quantum wire. Let us treat a single ring of tensors as a matrix product operator (MPO) acting on correlation space. Measurement in the $X$ basis destroys the symmetry \eqref{eqf:cluster_symms_z}, because it anticommutes with the measurement basis. The action of the remaining symmetries \eqref{eqf:cluster_symms_x} on the virtual legs define a local, unitary update rule, that is, a QCA~(Fig. \ref{fig:cols}a and \ref{fig:cols}b). Performing an $X$ measurement on every spin in the bulk of an $n\times d$ cluster state is equivalent to applying this QCA $d$ times, which entangles every qubit on the left edge with every qubit on the right. In this way, we obtain maximal transmission capacity.  

\subsubsection{GHZ state}

The Greenberger-Horne-Zeilinger (GHZ) state~\cite{Greenberger1989} is important in the study of quantum foundations~\cite{Greenberger1990}, entanglement~\cite{Dur2000,Sanz2016}, and SPT phases~\cite{Schuch2011}, however, it is of limited computational use. With a GHZ state, the quantum wire protocol creates only one qubit worth of entanglement, $C(n,d)=1$ regardless of the circumference. Consequently, there is only one logical qubit available for the more complex task of MBQC. Furthermore, the only operations available for this logical qubit are $Z$ rotations. 

As a stabilizer PEPS, the local tensors of the GHZ state satisfy
\begin{equation} \label{eqf:ghz x}
  \includegraphics[valign=c, scale=0.7]{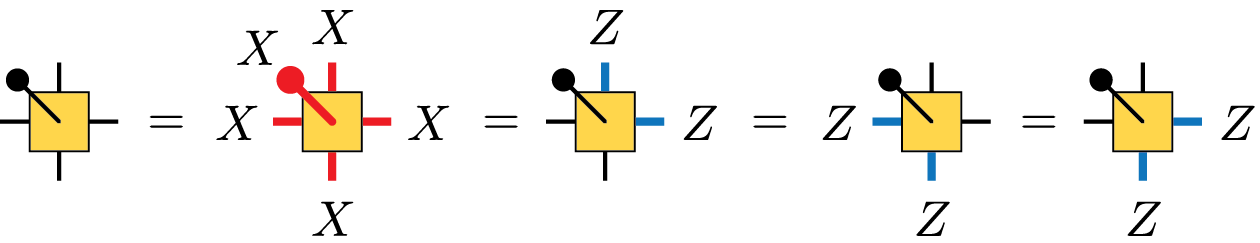}
\end{equation}
\begin{equation} \label{eqf:ghz z}
  \includegraphics[valign=c, scale=0.7]{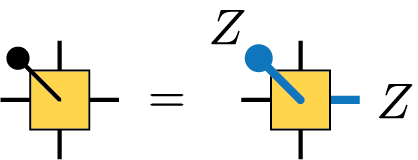}
\end{equation}
With \eqref{eqf:ghz x} and \eqref{eqf:ghz z} we may construct the stabilizers of the GHZ state (Fig. \ref{fig:stabs}b and \ref{fig:stabs}c), just as the cluster state stabilizers are constructed from \eqref{eqf:cluster_symms_x} and \eqref{eqf:cluster_symms_z}. However, upon forming a ring of GHZ tensors we find that the resulting MPO is not a QCA because it has a non-trivial kernel, i.e. some operators are mapped to the identity when reading from left to right (Fig. \ref{fig:cols}c). The image of the MPO is spanned by just two operators, which permits transmission of a single qubit (Fig. \ref{fig:cols}e and \ref{fig:cols}f). 

\begin{figure}
    \centering
    \begin{tabular}{ccc}
        \includegraphics[scale=0.7]{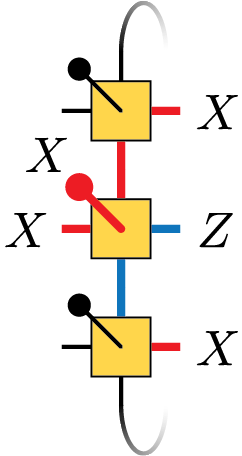} \hspace{5mm} & \includegraphics[scale=0.7]{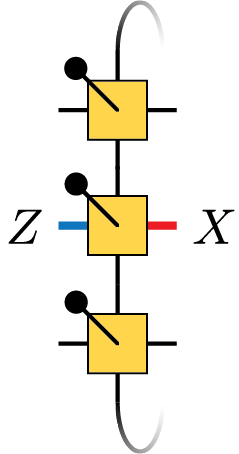} \hspace{5mm} & \includegraphics[scale=0.7]{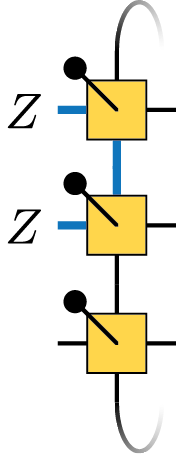} \\
        \rule[-5mm]{0pt}{9mm} (a) & (b) & (c) \\
        \includegraphics[scale=0.7]{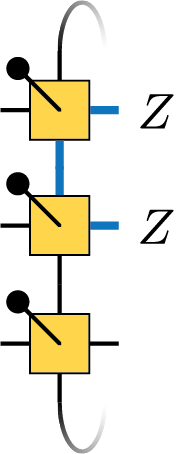} \hspace{5mm} & \includegraphics[scale=0.7]{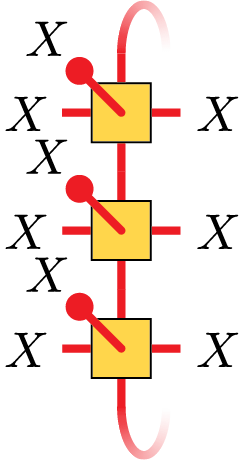} \hspace{5mm} & \includegraphics[scale=0.7]{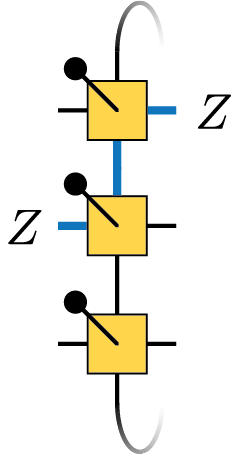} \\
        \rule{0pt}{4mm} (d) & (e) & (f)
    \end{tabular}
    \caption{When measured in the $X$ basis, a single layer of a cylindrical stabilizer PEPS constitutes an MPO acting on correlation space. The MPO acts to the left, thus information flows from left to right. (a,~b) The symmetry action on the virtual legs of the cluster state tensors defines a QCA with update rule $T(X_i) = X_{i-1}Z_iX_{i+1}$ and $T(Z_i)=X_i$. (c-f) The symmetries of the GHZ state define a non-unitary update rule.}
    \label{fig:cols}
\end{figure}

\subsubsection{Toric code}

The toric code~\cite{Kitaev2003} is a popular quantum code, suited to realizing fault-tolerant universal quantum computation in planar architectures~\cite{Raussendorf2007}. All states in the toric code space have topological order.

In our quantum wire setting, the toric code state has transmission capacity $C(n,d)=n-1$. In the context of MBQC we can perform some unitary operations, however, it has been shown that MBQC on the toric code can be simulated efficiently on a classical computer~\cite{Bravyi2007,Bravyi2022}. We are not aware of any work on a computational phase based around the toric code; whether or not such a phase exists remains an open question.

The toric code appears in the class of stabilizer PEPS. Consider a stabilizer tensor with the following symmetries:
\begin{equation} \label{eqf:toric_x}
  \includegraphics[valign=c, scale=0.7]{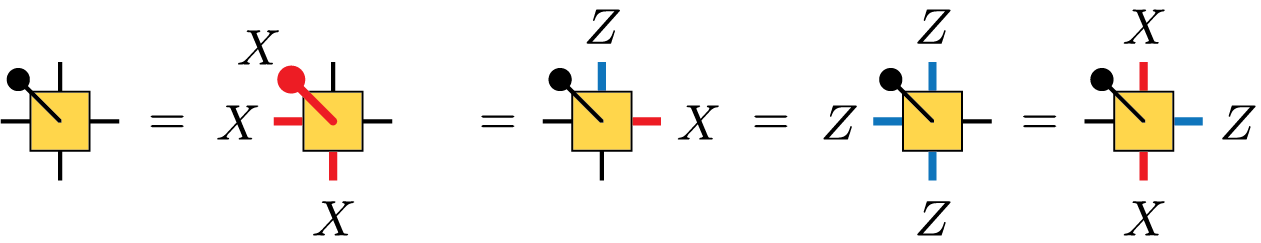}
\end{equation}
\begin{equation} \label{eqf:toric_z}
  \includegraphics[valign=c, scale=0.7]{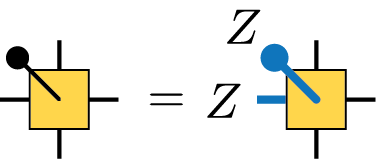}
\end{equation}
From \eqref{eqf:toric_x} and \eqref{eqf:toric_z}, we construct the 4-local stabilizer generators shown in Fig. \ref{fig:stabs}d. By shearing the lattice, it can be seen that these correspond to a translation-invariant version of the toric code, also known as the XZZX code~\cite{wenQuantumOrdersExact2003, kayCapabilitiesPerturbedToric2011}, obtained by applying Hadamard operators along every other column of the original toric code.    

A ring of toric code tensors gives an MPO that is non-unitary, similar to the GHZ tensors. The difference is that the kernel of the toric code MPO is smaller in dimension, which allows a greater quantity of quantum information to pass through each layer of the cylinder. 

\begin{figure*}[!htbp]
  \centering
  \includegraphics[width=\textwidth]{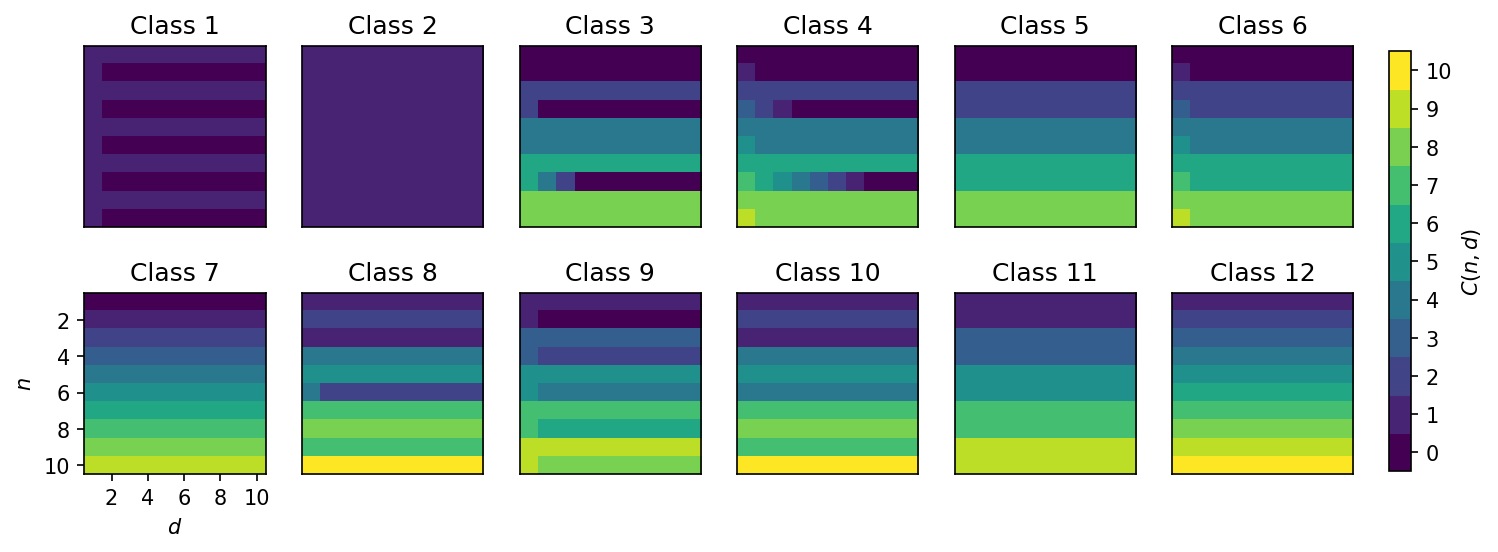}
  \caption{Transmission capacity $C(n,d)$ for a stabilizer PEPS falls into one of 13 classes. Class 12 is the class of QCA, $C(n,d)=n$. Class 0 (not shown) is the trivial class, $C(n,d)=0$.}
  \label{fig:heatmaps}
\end{figure*}

The representability of states in the toric code by PEPS is in no contradiction to their topological order. When states in the toric code are created unitarily from a tensor product state, long-range action is required. However, when projective measurement and classical communication are allowed, short-range quantum operations suffice. For example, states in the toric code can be obtained by individually measuring every other qubit in a 2D cluster state. A systematic classification of quantum states under local quantum operations and classical communication, albeit only in spatial dimension 1, has been provided in~\cite{Piroli2021}.

\subsection{Results} \label{subsec:results}

As illustrated by the above examples, the class of stabilizer PEPS contains states with significantly different capabilities for quantum information processing. Focusing on quantum wire, we immediately discern that the QCA-type maximally transmitting channel found for the cluster state is only one possible behavior. In other cases, the transmission capacity may be smaller -- as little as one qubit, independent of the circumference and distance. Our goal is to classify the transmission behaviors that arise within the class of stabilizer PEPS. Specifically, we ask: \textit{given a stabilizer PEPS, how many qubits worth of quantum information can be transmitted by local measurement, as a function of circumference and depth?} The answer to this question is given by the following theorem.

\begin{theorem} \label{thm:wire}
  For the class of cylindrical stabilizer PEPS on a 2D square lattice, there are 13 distinct behaviors of quantum wire transmission, characterized by the transmission capacity $C(n,d)$ as a function of the circumference $n$ and the depth $d$. These include the trivial class $C(n,d)=0$ and the class corresponding to Clifford quantum cellular automata, $C(n,d)=n$. The other eleven classes are represented in Fig.~\ref{fig:heatmaps}.
\end{theorem}

There are 2649 distinct stabilizer PEPS up to gauge equivalence, thus Theorem \ref{thm:wire} is a significant simplification to the phenomenology of stabilizer PEPS.

\section{Classification of Quantum Wire} \label{sec:classification}

The purpose of this section is to prove Theorem \ref{thm:wire}. The proof splits into two parts: one analytical and one numerical. First, we reduce the size of our problem considerably with the following proposition.
\begin{restatable}{prop}{myprop} \label{prop:wire}
  Let $C(A, n, d)$ be the transmission capacity of a cylindrical stabilizer PEPS constructed from local tensor $A$ with circumference $n$ and depth $d$. If $C(A_1,n,d)=C(A_2,n,d)$ for $n\leq 6$ and $d\leq 6$, then $C(A_1,n,d)=C(A_2,n,d)$ for all values of $n$ and $d$.
\end{restatable}
With Proposition \ref{prop:wire} in hand, the proof of Theorem \ref{thm:wire} is a direct consequence of the following lemma.
\begin{lemma}
  For $n\leq 6$ and $d\leq 6$, the transmission capacity $C(n,d)$ must take one of the 12 forms represented in Fig. \ref{fig:heatmaps}, or the trivial form $C(n,d)=0$. 
\end{lemma}
\begin{proof}
 The proof is by computer. Code to reproduce our results is available at~\cite{2dwire}.
\end{proof}
Below, we set out to prove Proposition \ref{prop:wire}.

\subsection{Background, notation, and definitions} \label{subsec:background}

The single-qubit Pauli group $\mathcal{P}_1$ is the group of observables generated by the Pauli $X$ and $Z$ operators: $\mathcal{P}_1 = \expval{X, Z}$. The $n$-qubit Pauli    group is the $n$-fold tensor product of $\mathcal{P}_1$. An $n$-qubit stabilizer group is an abelian subgroup $S\subset \mathcal{P}_n$ such that $-I \notin S$, and an $n$-qubit stabilizer state is an $n$-qubit quantum state uniquely specified by a stabilizer group. Mappings between stabilizer states correspond to elements of the Clifford group, which is the normalizer of $\mathcal{P}_n$ in $U(2^n)$, defined here modulo a complex phase: $\mathcal{C}_n = \{U:U\mathcal{P}_nU^\dagger = \mathcal{P}_n\}/U(1)$. 

We also make use of the symplectic representation of the Pauli group by mapping elements of $\mathcal{P}_n$ to binary vectors $\xi = (\xi^X, \xi^Z) \in \mathbb{F}_2^{2n}$, such that $g(\xi) = \bigotimes_{i=1}^n X^{\xi_i^X}Z^{\xi_i^Z} \in \mathcal{P}_n$~\cite{Stephen2019}. Multiplication of Pauli operators corresponds to addition of binary vectors, and commutation relations are captured by a symplectic form $\omega(\xi, \eta) = \sum_{i=1}^n \xi^X_i \eta^Z_i + \xi^Z_i \eta^X_i$ modulo 2 such that $g(\xi)g(\eta) = (-1)^{\omega(\xi,\eta)}g(\eta)g(\xi)$. The vector space $\mathbb{F}_2^{2n}$ equipped with the symplectic form $\omega$ is called a symplectic vector space, and it is isomorphic to $\mathcal{P}_n$ up to phase, i.e., $\mathbb{F}_2^{2n} \cong \mathcal{P}_n/U(1)$. A stabilizer group corresponds to an isotropic subspace of $\mathbb{F}_2^{2n}$ with respect to $\omega$, and a stabilizer state to a maximally isotropic subspace.

Consider the $\ell$-qubit stabilizer states. There is a natural mapping from these states to PEPS tensors with $k$ physical legs of dimension $D=2$ and $\ell-k$ virtual legs of dimension $\chi=2$. Namely, we assign each qubit of the stabilizer state to one leg of the tensor and re-frame the stabilizer group as the symmetry group of the tensor. This completely specifies the components of the tensor in the same way that a stabilizer state is completely specified by its symmetries. 
\begin{definition}
  A \textbf{stabilizer tensor} is a PEPS tensor whose components are completely specified by a set of symmetry constraints, such that the same symmetry constraints also specify a stabilizer state. An $[\ell,k]$ stabilizer tensor has $\ell$ legs in total and $k$ physical legs. A \textbf{stabilizer PEPS} is a translation-invariant PEPS where the local tensor is a stabilizer tensor. 
\end{definition}

In this paper we work with a square 2D lattice, so we consider the $[5,1]$ stabilizer tensors. Two stabilizer tensors $A_1$ and $A_2$ generate the same stabilizer PEPS if they are related by a gauge transformation:
\begin{equation}
  \includegraphics[valign=c, scale=0.75]{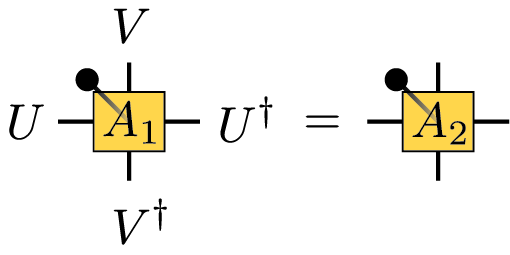}.
\end{equation}
where $U, V \in \mathcal{C}_1$. Up to gauge equivalence, there are 2649 distinct stabilizer PEPS.

We denote a cylindrical stabilizer PEPS $\mathcal{A}(A,n,d)$ with local stabilizer tensor $A$, circumference $n$, and depth $d$. Upon measurement of every physical spin in the $X$ basis, we are left with a stabilizer state on the $2n$ virtual legs at the edges. The stabilizer group for this state is $S(A,n,d)$. $S$ has a natural bipartition into the two open edges of the cylinder; we write $S_L$ and $S_R$ to denote the restriction of $S$ to the left and right edges, respectively. 

The generators of a bipartite stabilizer group such as $S(A,n,d)$ can always be brought to the following canonical form, which reveals the entanglement structure of the corresponding stabilizer state~\cite{Fattal2004}:
\begin{equation} \label{eq:canonical stabs}
 \begin{aligned}
 \mathcal{Z}_L(A,n,d) 
 &= 
 \begin{cases}
 a_i \otimes I \\
 \end{cases} \qquad 1 \leq i \leq n-p \\
 P_{LR}(A,n,d) &=   
 \begin{cases}
 g_k^L \otimes g_k^R \\
 \bar{g}_k^L \otimes \bar{g}_k^R
 \end{cases} \qquad 1 \leq k \leq p \\
 \mathcal{Z}_R(A,n,d) 
 &= 
 \begin{cases}
 I \otimes b_j \\
 \end{cases} \qquad 1 \leq j \leq n-p.
 \end{aligned}
\end{equation}
Here $\mathcal{Z}_L$ is the center of $S_L$ and acts trivially on partition $R$, while $\mathcal{Z}_R$ is the center of $S_R$ and acts trivially on partition $L$. $P_{LR}$ is generated by $p$ locally anticommuting pairs, meaning that $(g_k^L, \bar{g}_k^L)$ anticommute with each other but commute with all other $g_\ell^L$ and $\bar{g}_\ell^L$ for $\ell\neq k$. The same relationship holds for $(g_k^R, \bar{g}_k^R)$, so that $S_{LR}(n,d)$ is abelian when we look at both partitions simultaneously. We write $P_L$ and $P_R$ for the restriction of $P_{LR}$ to the left and right partitions, respectively.

Later, it will be convenient to encapsulate all the information about the canonical form of $S(A,n,d)$ with a single mathematical object. To this end, we denote the canonical form of $S(A,n,d)$ by $\Phi_d(A,n)$. Every $\Phi_d$ with the same $(n,p)$ is equivalent to $p$ Bell pairs encoded by an $[n,p]$ stabilizer code. Therefore, the transmission capacity $C(A,n,d)$ is equal to $p$, and $\Phi_d$ naturally separate into equivalence classes labelled by the tuple $(n,p)$. 

\begin{definition}
  The \textbf{transmission capacity} $C(A,n,d)$ for a stabilizer PEPS $\mathcal{A}(A,n,d)$ with local stabilizer tensor $A$, depth $d$, and circumference $n$ is the number of anticommuting pairs $p$ in the canonical form of the bipartite stabilizer group $S(A,n,d)$.
\end{definition}

\subsection{Proof of Proposition \ref{prop:wire}} \label{subsec:proof}

In this section, we prove Proposition \ref{prop:wire}. Here, we provide a brief overview of the proof.

First, we generalize the update rule corresponding to a single layer of a stabilizer PEPS so that it can deal with non-unitary updates, like those found for the GHZ state and toric code. Then, we use our generalized update rule to analyze the concatenation of many individual layers, and show that $C(n,d)$ is completely determined by two mathematical objects, both of which may be derived from a single layer of the tensor network. These objects are the canonical stabilizer generators for a single layer, and the commutation relations between canonical stabilizer generators for a single layer. We denote them $\Phi_1$ and $\Omega_1$, respectively. 

Both $\Phi_1$ and $\Omega_1$ are independent of $d$ but remain functions of $n$; we remove this dependence by deriving a local generating set for the stabilizer group of a single layer. Relating the local stabilizer generators to the canonical stabilizer generators allows us to show that if there are two different stabilizer tensors $A_1$ and $A_2$ with $\Phi_1(A_1,n)=\Phi_1(A_2,n)$ and $\Omega_1(A_1,n)=\Omega_1(A_2,n)$ for $n\leq 6$, then $\Phi_1(A_1)=\Phi_1(A_2)$ and $\Omega_1(A_1)=\Omega_1(A_2)$ for all $n$. 

\subsubsection{Non-unitary update rule}

Consider the MPO $\mathcal{A}(A,n,1)$. For compactness of notation, we draw $\mathcal{A}$ as a single tensor with higher bond dimension:
\begin{equation} \label{eqf:quasi 1d}
 \includegraphics[valign=c, scale=0.75]{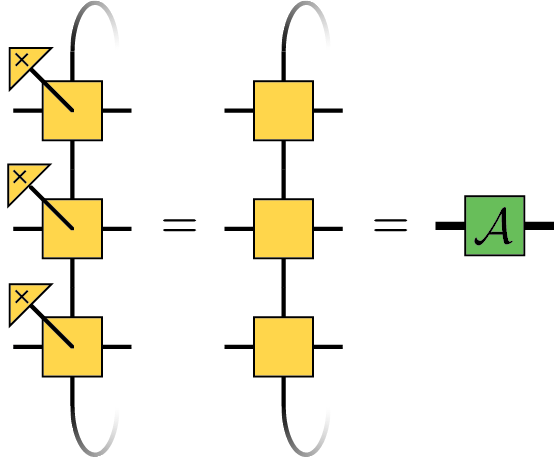}.
\end{equation}
Note that we have removed the physical legs of every tensor, because measurement in the $X$ basis contracts these legs with an $X$ eigenstate. Strictly speaking, the action of $\mathcal{A}$ depends on the measurement outcomes at each site. However, the only effect that different measurement outcomes may have on the stabilizer generators for $\mathcal{A}$ is to change some of their signs. This change does not affect the number of anticommuting pairs in the canonical form \eqref{eq:canonical stabs}, i.e., it does not affect the transmission capacity. Therefore, without loss of generality we contract every physical leg with the $\ket{+}$ state. To be clear, any tensor drawn without a physical leg is assumed to be contracted with $\ket{+}$.

The stabilizer group of $\mathcal{A}(A,n,1)$ is $S(A,n,1)$. Bringing $S$ to canonical form gives the following symmetries for $\mathcal{A}$:
\begin{equation} \label{eqf:canonical_symms}
 \includegraphics[valign=b, scale=0.75]{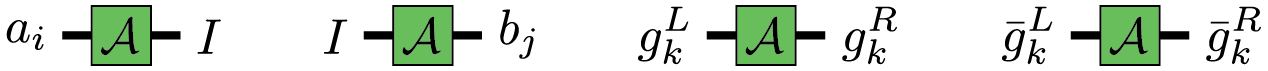}.
\end{equation}
In the spirit of QCA, we use Eq. \ref{eqf:canonical_symms} to define an update rule $T$ such that
\begin{equation} \label{eq:T_nonunitary}
 T(a_i) = I; \quad T(g_k^L) = g_k^R; \quad T(\bar{g}_k^L) = \bar{g}_k^R.
\end{equation}
$T$ is non-unitary unless every generator of $S$ is part of an anticommuting pair, in which case we recover a QCA by identifying $g_k^L$ and $\bar{g}_k^L$ with encoded Pauli operators $\tilde{X}_k$ and $\tilde{Z}_k$, respectively. 

Unitarity for any value of $p$ can be restored by modifying the update rule such that
\begin{equation} \label{T_unitary}
 T(a_i) = b_i, \quad T(g_k^L) = g_k^R; \quad T(\bar{g}_k^L) = \bar{g}_k^R.
\end{equation}
However, note that $T$ is not defined on the full $n$-qubit Hilbert space. Instead, it maps the subspace $S_L$ to the subspace $S_R$. If $S_L\neq S_R$, then some operators in $S_L$ have an image under $T$ that falls outside the domain of $T$. A simple test to determine whether an operator is in the domain of $T$ follows directly from the following lemma. 
\begin{lemma} \label{lem:centralizer}
 $S_L(A,n,d)$ is the centralizer of $\mathcal{Z}_L(A,n,d)$ in $\mathcal{P}_n$, and $S_R(A,n,d)$ is the centralizer of $\mathcal{Z}_R(A,n,d)$ in $\mathcal{P}_n$.
\end{lemma}
\begin{proof}
 Without loss of generality we prove the statement for $S_L$ and $\mathcal{Z}_L$, as the same argument applies to the right partition. Switching to the symplectic representation of $\mathcal{P}_n$, we have $\mathcal{Z}_L$ an isotropic subspace and $S_L \subseteq \mathcal{Z}_L^\perp$. By Eq. \ref{eq:canonical stabs}, the dimensions of these subspaces are $\dim(\mathcal{Z}_L) = n-p$ and $\dim(S_L) = n+p$. 
 
 For any symplectic vector space $V$ and subspace $W \subset V$ it is known that $\dim(V) = \dim (W) + \dim (W^\perp)$. Adding dimensions, we have $\dim(\mathcal{Z}_L) + \dim(S_L) = 2n = \dim(\mathcal{P}_n)$. Therefore $S_L = \mathcal{Z}_L^\perp$. In other words, $S_L$ centralizes $\mathcal{Z}_L$ in $\mathcal{P}_n$.
\end{proof}
\begin{corollary} \label{cor:T_image_domain}
  For any $x \in \mathcal{P}_n$, $x \in S_L(A,n,d)$ if and only if $x$ commutes with $\mathcal{Z}_L(A,n,d)$.
\end{corollary}
\begin{proof}
  $\mathcal{Z}_L$ is the center of $S_L$, therefore $x \in S_L$ implies that $x$ commutes with every $y\in\mathcal{Z}_L$. By Lemma \ref{lem:centralizer}, $S_L$ is the centralizer of $\mathcal{Z}_L$ in $\mathcal{P}_n$, therefore the converse is also true.
\end{proof}

\subsubsection{Chains of symmetry}

Any symmetry of $\mathcal{A}(A,n,d)$ can be derived by chaining symmetries, i.e., concatenating $d$ copies of $\mathcal{A}(A,n,1)$. We use the chaining process to characterize the group $\mathcal{Z}_R(A,n,d)$, because one way to determine the channel capacity is $C(A,n,d) = n - \rank[\mathcal{Z}_R(A,n,d)]$. 

Elements of $\mathcal{Z}_R(A,n,d)$ must act as the identity on the left edge of $\mathcal{A}(A,n,d)$ and non-trivially on the right edge. Therefore, they have the following form: 
\begin{equation} \label{eqf:Z_chain_symms}
    \includegraphics[valign=c, scale=0.75]{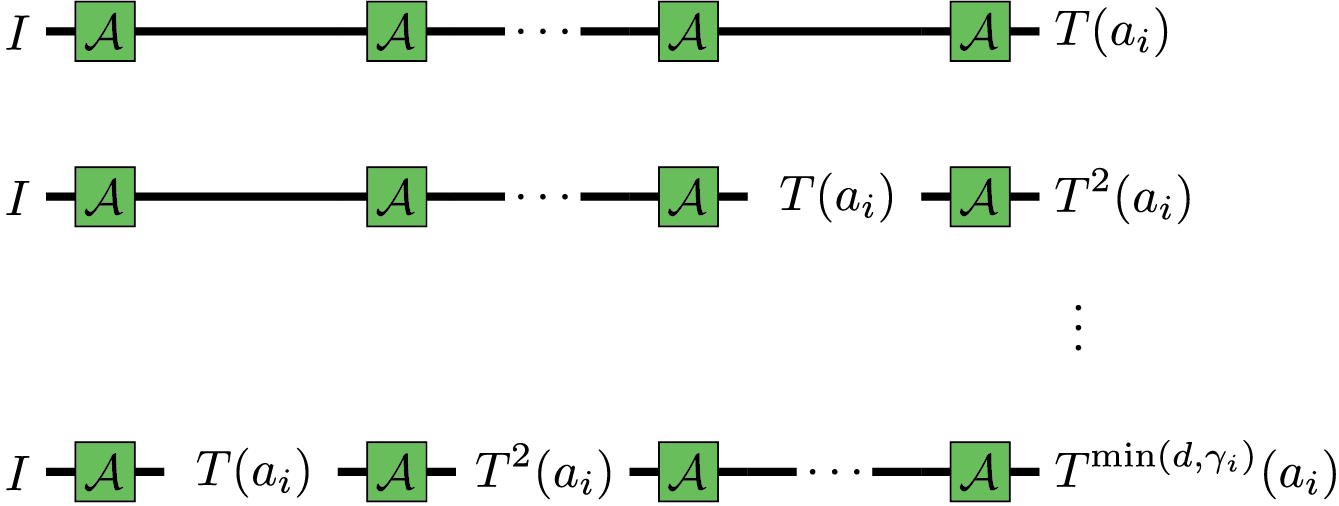}
\end{equation}
where $a_i \in \mathcal{Z}_L(A,n,1)$, $1\leq i \leq n-p$, and $\gamma_i$ is the critical depth such that $T^{\gamma_i}(a_i) \notin S_L(A,n,1)$. Recall that this implies $T^{\gamma_i+1}(a_i)$ is not a symmetry. Reading off the rightmost column of \eqref{eqf:Z_chain_symms}, it is clear that $\mathcal{Z}_R(A,n,d)$ is generated by the set $\mathcal{T}_R(n,d) = \qty{T(a_i), \:\dots, \:T^{\min(d,\gamma_i)}(a_i) \:|\: 1\leq i \leq n-p}.$ Note that $\mathcal{Z}_R(A,n,d)$ is invariant for $d \geq \max_i (\gamma_i)$, because every chain in \eqref{eqf:Z_chain_symms} may be extended to the left by the trivial symmetry, but no chain may be extended to the right once $d \geq \max_i(\gamma_i)$.

\begin{lemma} \label{lem:step_down}
  For any fixed value of $n$, $C(A,n,d) \geq C(A,n,d+1)$.
 \end{lemma}
 \begin{proof}
  $\mathcal{T}_R(A,n,d)$ generates $\mathcal{Z}_R(A,n,d)$, therefore 
  \begin{equation} \label{eq:C from TR}
  C(A,n,d) = n - \rank [\mathcal{T}_R(A,n,d)].
  \end{equation}
  Any of the chains \eqref{eqf:Z_chain_symms} may be extended indefinitely to the left by the trivial symmetry, therefore $\mathcal{T}_R(A,n,d) \subseteq \mathcal{T}_R(A,n, d+1)$ and $\rank [\mathcal{T}_R(A,n,d)] \leq \rank [\mathcal{T}_R(A,n,d+1)]$. Substituting \eqref{eq:C from TR} into this inequality gives $C(A,n,d) \geq C(A,n,d+1)$.
 \end{proof}
 \begin{lemma} \label{lem:plateau}
  For any fixed value of $n$, if $C(A,n,d_c) = C(A,n,d_c+1)$ at some critical depth $d_c$, then $C(A,n,d) = C(A,n,d_c)$ for all $d \geq d_c$. 
 \end{lemma}
 \begin{proof}
  $C(A,n,d_c) = C(A,n,d_c+1)$ requires that $T$ map $\mathcal{Z}_R(A,n, d_c)$ to itself, therefore $\mathcal{Z}_R(A,n, d_c)$ is fixed for $d \geq d_c$. Recall that $C(A,n,d) = n - \rank[\mathcal{Z}_R(A,n,d)]$. Therefore, $C(A,n,d)$ is also fixed for all $d \geq d_c$.
 \end{proof}
 \begin{lemma} \label{lem:max_depth}
  For any fixed value of $n$, the maximum value of $d_c$ is $C(A,n,1)$.
 \end{lemma}
 \begin{proof}
  By Lemmas \ref{lem:step_down} and \ref{lem:plateau}, $C(A,n,d)$ must decrease when $d$ is increased, up to the critical value $d_c$. If it ever fails to decrease, then we have reached $d_c$ and the transmission capacity must be constant thereafter. $C(A,n,d)$ is a non-negative quantity; therefore, it may decrease at most $C(A,n,1)$ times before reaching zero and becoming constant. 
 \end{proof}

\subsubsection{Codifying the update rule $T$}

Using the Lemmas \ref{lem:step_down}, \ref{lem:plateau}, and \ref{lem:max_depth}, we formulate an algorithm to compute $C(A,n,d)$. 
\begin{table}[H]
    \centering
    \bgroup
    \def\arraystretch{1.5}
    \begin{tabular}{p{1.5ex}p{0.9\linewidth}}
        \hline \multicolumn{2}{p{0.8\linewidth}}{\textbf{Algorithm 1:} Transmission Capacity} \\ \hline
        \multicolumn{2}{p{0.8\linewidth}}{\textbf{Data:} Stabilizer tensor $A$, circumference $n$, depth $d$} \\
        \textbf{1} & Compute $\Phi_1(A,n)$, which determines $\mathcal{Z}_L(A,n,1)$ and the update rule $T$ \\
        \textbf{2} & Initialize $\mathcal{T}_R(A,n,1) = \{ T(x) | x \in G\}$, where $G$ is any generating set for $\mathcal{Z}_L(A,n,1)$ \\
        \textbf{3} & For every $y \in \mathcal{T}_R(A,n,1)$, check if $y$ commutes with $\mathcal{Z}_L(A,n,1)$. If so, $T(y)$ exists. Add every valid $T(y)$ to $\mathcal{T}_R(A,n,1)$ and obtain $\mathcal{T}_R(A,n,2)$ \\
        \textbf{4} & Repeat Step 3 until the desired depth is reached, or until $\mathcal{T}_R(A,n,d) = \mathcal{T}_R(A,n,d-1)$ \\
        \textbf{5} & Return $C(A,n,d) = n - \rank [\mathcal{T}_R(A,n,d)]$ \\ \hline
    \end{tabular}
    \egroup
\end{table}

Algorithm 1 shows that the following mathematical objects are sufficient to compute $C(A,n,d)$:
\begin{enumerate}[(i)]
 \item The canonical generators of $S(A,n,1)$, i.e. $\Phi_1(A,n)$.
 \item The commutation relations between generators of $\mathcal{Z}_L(A,n,1)$ and $S_R(A,n,1)$, which we denote $\Omega_1(A, n)$. 
\end{enumerate}
Recall that $\Phi_1$ is characterized by the tuple $(n,p)$, however, there is considerable freedom to define $\Omega_1$ by choosing different generators for $\Phi_1$. Therefore, it will be useful to have a standard form for $\Omega_1$. To this end, we represent $\Omega_1$ by a bipartite graph $G_\Omega$ with node sets $\{a_i\}$ and $\{b_j, g_k^R, \bar{g}_k^R\}$. Let the edges of $G_\Omega$ indicate anticommutation between operators in different node sets, and let the biadjacency matrix be $M$, with the following layout:
\begin{equation}
\kbordermatrix{& b_1 & \cdots & b_{n-p} & g_1^R & \bar{g}_1^R & \cdots & g_p^R & \bar{g}_p^R \\ a_1 \\ \vdots \\ a_{n-p}} 
\end{equation}
The operations on $M$ that preserve $\Phi_1$ are:
\begin{enumerate}[(i)]
 \item Swap any two rows, or any of the first $(n-p)$ columns. 
 \item Add any row to any other row, modulo 2. 
 \item Add any of the first $(n-p)$ columns to any of the last $p$ columns, modulo 2.
 \item Reorder the $2\times(n-p)$ blocks of $M_{\Omega_1}$ that represent the anticommuting pairs.
 \item Within any of the blocks for each anticommuting pair:
 \begin{enumerate}[(a)]
 \item Swap columns.
 \item Add one column to the other or vice versa, modulo 2.
 \end{enumerate}
\end{enumerate}
Using these operations, $M$ may be brought to the following standard form:
\begin{equation}
  \begin{array}{cccccccccc}
  \left(\rule{0cm}{1.2cm}\right. \mqty{\lambda_1 & & \\ & \ddots & \\ & & \lambda_{n-p}} & \left|\rule{0cm}{1.2cm}\right. & H_1 & \left|\rule{0cm}{1.2cm}\right. & \cdots & \left|\rule{0cm}{1.2cm}\right. & H_p \left.\rule{0cm}{1.2cm}\right)
  \end{array}
\end{equation}
where $\lambda_i \in \{0,1\}$ and $H_j$ are $(n-p)\times 2$ block matrices in RCEF, ordered by rank. The RREF and RCEF of any matrix (or block matrix) is unique, therefore this standard form is well-defined and we can use it to compare two stabilizer PEPS.

\subsubsection{Local generators}

The next step is to eliminate the dependence of $\Phi_1$ and $\Omega_1$ on $n$. We do this by exploiting the translation invariance of stabilizer PEPS to find a local form for the generators of $S(A,n,1)$. To be precise, we show that the generating set for $S(A,n,1)$ has a form where every generator save one is $k$-local with $k \leq 3$. Because there is only one nonlocal generator, we slightly abuse terminology and refer to the local form of $S(A,n,1)$. 

To obtain a local form of $S(A,n,1)$, we recursively apply \eqref{eq:canonical stabs} to the symmetry group of the local tensor $A$, using smaller and smaller partitions each time. Let the 4 virtual legs of $A$ be partitioned into horizontal and vertical subsystems $H$ and $V$, which are subdivided into $\{u, d, l, r\} = \{\text{up}, \text{down}, \text{left}, \text{right}\}$:
\begin{equation} \label{eqf:tensor partitions}
    \includegraphics[valign=c, scale=1.0]{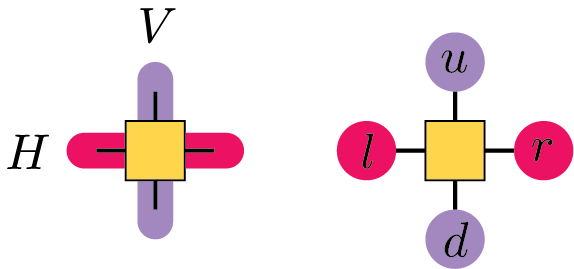}
\end{equation}
The first step in the recursive application of \eqref{eq:canonical stabs} is to bring $S(A,1,1)$ to canonical form with respect to the $\{H,V\}$ partition. For convenience, we denote $S(A,1,1)$ by $S_A$. The subsystems $H$ and $V$ have two legs each, so there may be $p=0$, 1, or 2 anticommuting pairs. The canonical form of $S_A$ for each case is:
\begin{equation} \label{eq:HV canonical form}
 \begin{aligned}[c]
 p &= 0 \\
 h_1 &\otimes I \\
 h_2 &\otimes I \\
 I &\otimes v_1 \\
 I &\otimes v_2
 \end{aligned}
 \qquad
 \begin{aligned}[c]
 p &= 1 \\
 h &\otimes I \\
 g^H &\otimes g^V \\
 \bar{g}^H &\otimes \bar{g}^V \\
 I &\otimes v
 \end{aligned}
 \qquad
 \begin{aligned}[c]
 p &= 2 \\
 g_1^H &\otimes g_1^V \\
 g_2^H &\otimes g_2^V \\
 \bar{g_1}^H &\otimes \bar{g_1}^V \\
 \bar{g_2}^H &\otimes \bar{g_2}^V.
 \end{aligned}
\end{equation}
Note that symmetries like $h_1\otimes I$ decouple the action of $S_A$ on the horizontal legs from its action on the vertical legs. $S(A,n,1)$ automatically inherits these symmetries for any $n$ because tiling them into a ring becomes trivial. In contrast, symmetries like $g^V\otimes g^H$ induce a coupling between the vertical and horizontal legs, in which case we cannot take the vertical tiling of $A$ for granted. 

\begin{lemma} \label{lem:HV p=0}
  If there are $p=0$ anticommuting pairs in the canonical form of $S_A$, then $S(A,n,1)$ has the following local generating set:
  \begin{equation}
  S(A,n,1) = \expval{h_{1,i}, \: h_{2,i}\:|\:i=1,\dots,n}.
  \end{equation}
 \end{lemma}
 \begin{proof}
  When $p=0$, the action of $S_A$ on the horizontal legs is decoupled from its action on the vertical legs, and $S(A,n,1)$ directly inherits the operators $h_{1,i}$ and $h_{2,i}$ for every site $i$. The set $\qty{h_{1,i}, \: h_{2,i}\:|\:i=1,\dots,n}$ is abelian with rank $2n$, therefore it is a generating set for $S(A,n,1)$. 
 \end{proof}

 When $p=1$, the situation is more complicated because the action of $S_A$ on the horizontal legs is not decoupled from its action on the vertical legs. Consequently, the elements of $S(A,n,1)$ depend on how the elements of $S_A$ match up on the vertical legs. To determine these vertical tilings, consider the restriction of $S_A$ to subsystem $V$, denoted $S_A^{(V)} = \expval{v, g^V, \bar{g}^V}.$
 
 By definition, $v$ and $g^V$ commute, thus the subgroup $R_A = \expval{v, g^V} \subset S_A^{(V)}$ is a 2-qubit stabilizer group and it may be brought to canonical form with respect to the $\{u,d\}$ partition. Depending on the number of anticommuting pairs $q$ in $R_A$, we have the following generating sets.
 \begin{equation} \label{eq:RA}
  \begin{aligned}
  &q=0: \quad R_A = \expval{\mu \otimes I, I \otimes \delta} \\
  &q=1: \quad R_A = \expval{\gamma^u \otimes \gamma^d, \bar{\gamma}^u \otimes \bar{\gamma}^d},
  \end{aligned}
 \end{equation}
 where Greek letters are used to differentiate operators which act on $\{u,d,l,r\}$ from those that act on $\{H,V\}$, and operators that differ only by an overbar must anticommute. 
 
 \begin{lemma} \label{lem:HV p=1 q=0}
  If there are $p=1$ anticommuting pairs in the canonical form of $S_A$ and $q=0$ anticommuting pairs in the canonical form of $R_A$, then $S(A,n,1)$ has the following local generating set:
  \begin{equation}
  S(A,n,1) = \expval{g^H_i, \: h_i \:|\:i=1,\dots,n}.
  \end{equation}
 \end{lemma}
 \begin{proof}
  Without loss of generality let $v = \mu \otimes I$, so that $S_A^{(V)} = \expval{\mu\otimes I, I\otimes \delta, I\otimes \bar{\delta}}.$ Because $v$ commutes with $g^V$ and $\bar{g}^V$, we are free to multiply either of the latter by $v$ without altering the commutation relations. To account for this fact, we deal with the quotient group $S_A / \{e,v\}$. The cosets of $\{e, v\}$ are as follows.
  \begin{equation} \label{eq:HV p=1 q=0 partial cosets}
  \begin{aligned}
  &\qty[g^V]_v = \{I \otimes \delta, \mu \otimes \delta\} \\
  &\qty[\bar{g}^V]_v = \{I\otimes \bar{\delta}, \mu\otimes\bar{\delta}\} \\
  &\qty[g^V \bar{g}^V]_v = \{I\otimes \delta\bar{\delta}, \mu\otimes\delta\bar{\delta}\}.
  \end{aligned}
  \end{equation}
  Note that $\{\delta,\bar{\delta},\delta\bar{\delta}\} \in \{X,Y,Z\}$ are independent, giving one coset for each Pauli operator. Clearly, one of $\delta$, $\bar{\delta}$, or $\delta\bar{\delta}$ must equal $\mu$. Therefore, without loss of generality let us choose $g^V = I\otimes \mu$, and write every element of $S_A^{(V)}$ in terms of $\mu$ and its anticommuting partner $\bar{\mu}$:
  \begin{equation} \label{eq:HV p=0 q=0 cosets}
  \begin{aligned}
  &\qty[v]_v = \{\mu\otimes I, I\otimes I\} \\
  &\qty[g^V]_v = \{I \otimes \mu, \mu \otimes \mu\} \\
  &\qty[\bar{g}^V]_v = \{I\otimes \bar{\mu}, \mu\otimes\bar{\mu}\} \\
  &\qty[g^V \bar{g}^V]_v = \{I\otimes \mu\bar{\mu}, \mu\otimes\mu\bar{\mu}\}.
  \end{aligned}
  \end{equation} 
  
  Recall that $S(A,n,1)$ is constructed from $S_A$ by tiling elements of $S_A$ such that they match on contracted tensor legs. Using \eqref{eq:HV p=0 q=0 cosets} it is clear that the tiling $g^V_i v_{i+1}$ is valid, because $\mu$ is matched:
  \begin{equation} \label{eqf:HV p=0 q=0 tiling}
  \includegraphics[valign=c]{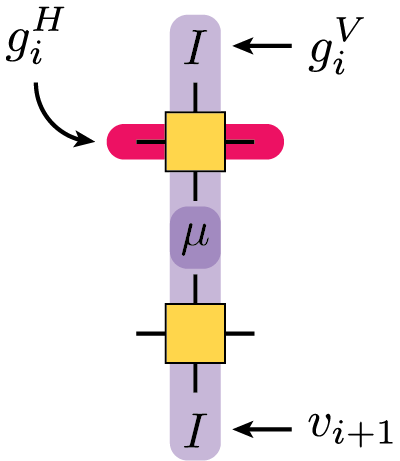}
  \end{equation}
  We have $g^V_i$ coupled to $g^H_i$ and $v_i$ not coupled to anything on the horizontal partition, therefore $g^H_i \in S(A,n,1)$ for all $i$, as shown in Eq. \ref{eqf:HV p=0 q=0 tiling}. This gives $n$ local generators of $S(A,n,1)$. The remaining $n$ local generators are given by $h_i$, which are trivially elements of $S(A,n,1)$ because they are decoupled from $V$. 
\end{proof}
\begin{restatable}{lemma}{mylemma} \label{lem:HV p=1 q=1}
  If there are $p=1$ anticommuting pairs in the canonical form of $S_A$ and $q=1$ anticommuting pairs in the canonical form of $R_A$, then the local generating set of $S(A,n,1)$ has one of two possible forms, depending on whether $\gamma^u = \gamma^d$ or $\gamma^u\neq\gamma^d$.
  \begin{enumerate}[(i)]
  \item If $\gamma^u=\gamma^d$, then $S(A,n,1) = \expval{h_i, \: g^H_ig^H_{i+1}, \: \prod_i \bar{g}^H_i\:|\:i=1,\dots,n}$ 
  \item If $\gamma^u\neq\gamma^d$, then $S(A,n,1) = \expval{h_i, \: \bar{g}^H_{i-1}g^H_i\bar{g}^H_{i+1}\:|\:i=1,\dots,n}$
  \end{enumerate}
\end{restatable}
\begin{proof}
  The proof is similar to Lemma \ref{lem:HV p=1 q=0}, see Appendix \ref{apx:p1q1}. 
\end{proof}
\begin{lemma} \label{lem:HV p=2}
  If there are $p=2$ anticommuting pairs in the canonical form of $S_A$, then $S(A,n,1)$ has the following local generating set: 
  \begin{equation}
  S(A,n,1)= \expval{g^H_{1,i}g^H_{2,i+1}, \: \bar{g}^H_{1,i}\bar{g}^H_{2,i+1} \:|\: i=1,\dots,n}.
  \end{equation}
 \end{lemma}
 \begin{proof}
  When $p=2$, the restriction of $S_A$ to $V$ is $S_A^{(V)} = \{g_1^V, g_2^V, \bar{g}_1^V, \bar{g}_2^V\}$, which has rank 4. Therefore, we must have $S_A^{(V)} = \mathcal{P}_2$. Without loss of generality, we may pick any generating set for $\mathcal{P}_2$ that matches the commutation structure of $S_A^{(V)}$. Therefore, let 
  \begin{equation}
  g^V_1 = I\otimes X, \quad g^V_2 = X\otimes I, \quad \bar{g}^V_1 = I\otimes Z, \quad \bar{g}^V_2 = Z\otimes I.
  \end{equation}
  Clearly the vertical tilings are $g^V_{1,i} g^V_{2,i+1}$ and $\bar{g}^V_{1,i} \bar{g}^V_{2,i+1}$. On the horizontal partition, these correspond to $g^H_{1,i} g^H_{2,i+1}$ and $\bar{g}^H_{1,i} \bar{g}^H_{2,i+1}$, respectively, which form a set of $2n$ independent, commuting elements. Therefore, they form a generating set for $S(A,n,1)$.
 \end{proof}
\begin{lemma} \label{lem:HV summary}
    The generating set for $S(A,n,1)$ may always be written in one of four local forms:
    \begin{enumerate}[(i)]
        \item $\qty{h_{1,i}, \: h_{2,i}\:|\:i=1,\dots,n}$
        \item $\qty{h_i, \: g^H_ig^H_{i+1}, \: \prod_i \bar{g}^H_i\:|\:i=1,\dots,n}$
        \item $\qty{h_i, \: \bar{g}^H_{i-1}g^H_i\bar{g}^H_{i+1}\:|\:i=1,\dots,n}$
        \item $\qty{g^H_{1,i}g^H_{2,i+1}, \: \bar{g}^H_{1,i}\bar{g}^H_{2,i+1} \:|\: i=1,\dots,n}$
    \end{enumerate}
\end{lemma}
\begin{proof}
    See Lemmas \ref{lem:HV p=0} through \ref{lem:HV p=2}. The local generating sets from Lemmas \ref{lem:HV p=0} and \ref{lem:HV p=1 q=0} are isomorphic, so they both fall under (i) in the above.
\end{proof}

\begin{figure}
    \centering
    \begin{tabular}{cc}
        \includegraphics[scale=0.6]{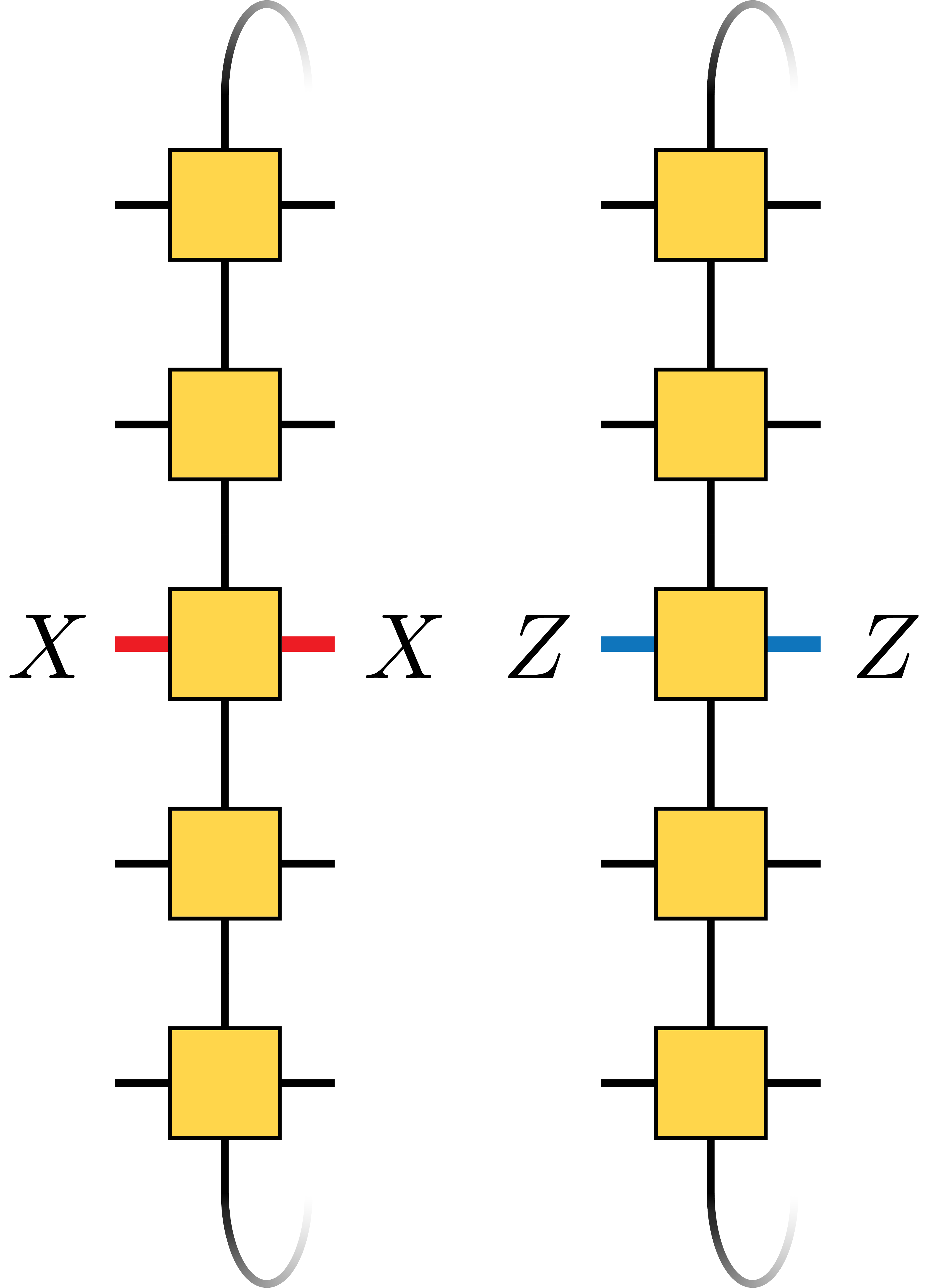} \hspace{5mm} & \includegraphics[scale=0.6]{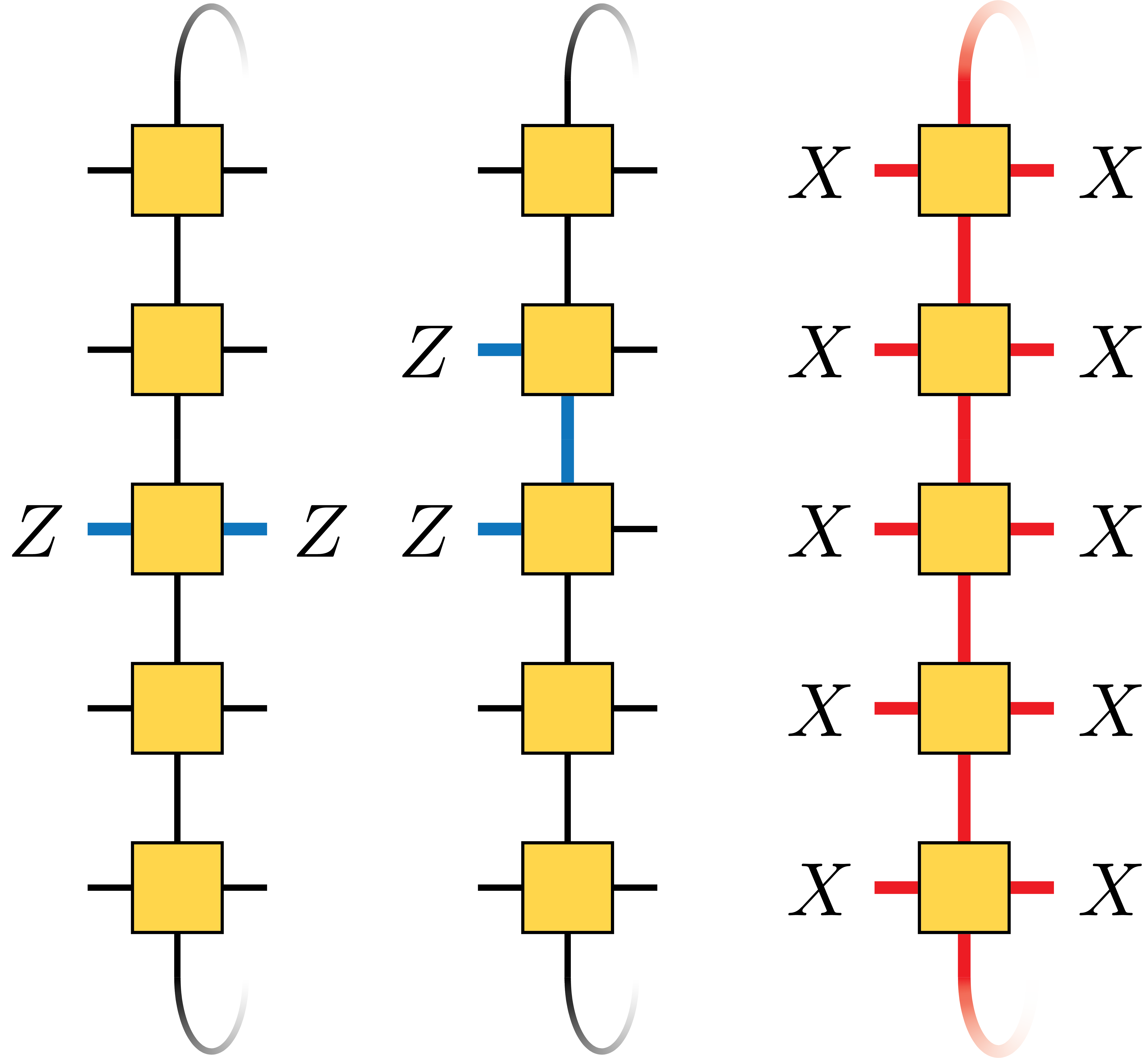} \hspace{5mm} \\
        \rule[-5mm]{0pt}{9mm} (i) & (ii) \\
        \includegraphics[scale=0.6]{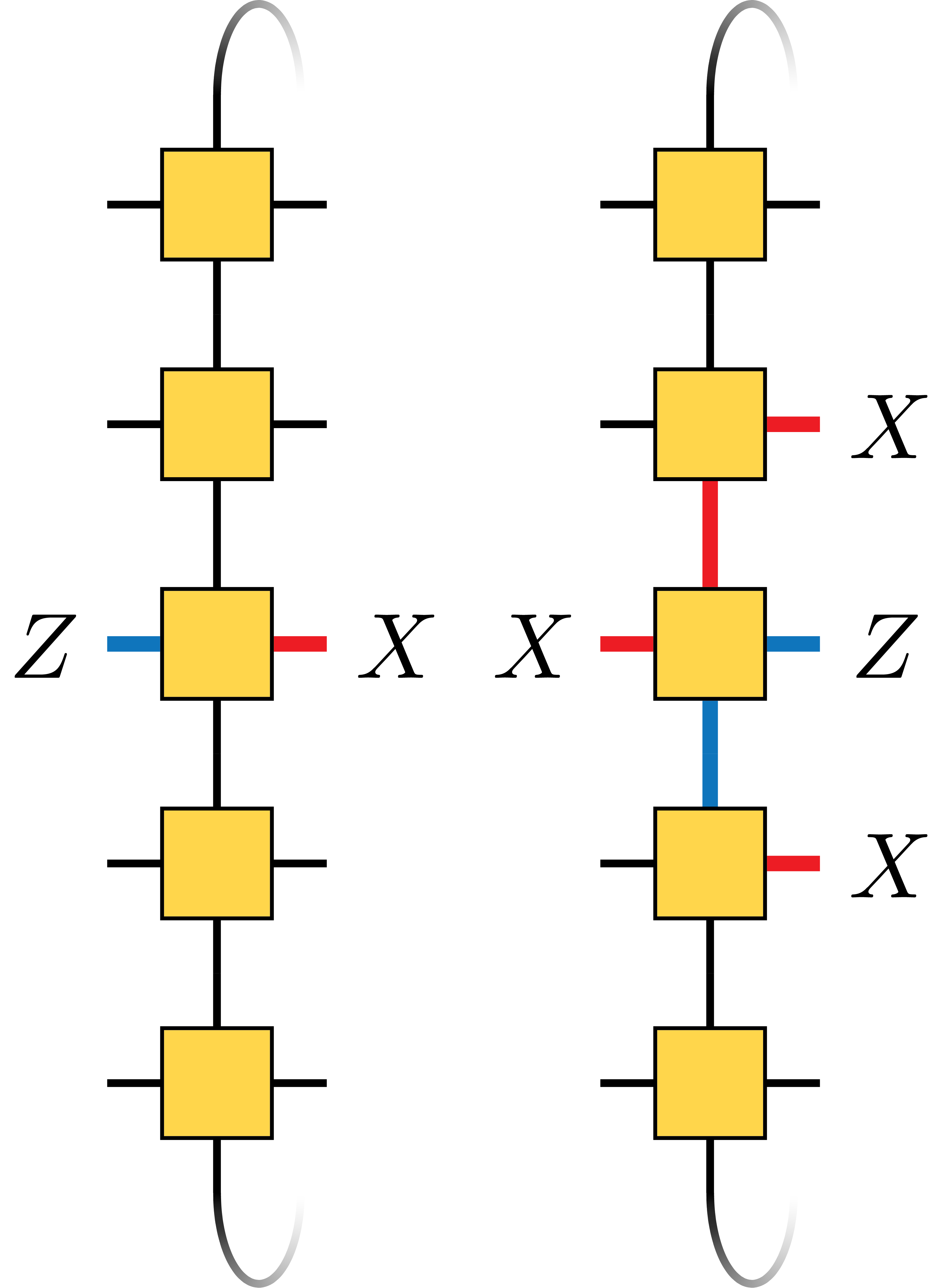} & \includegraphics[scale=0.6]{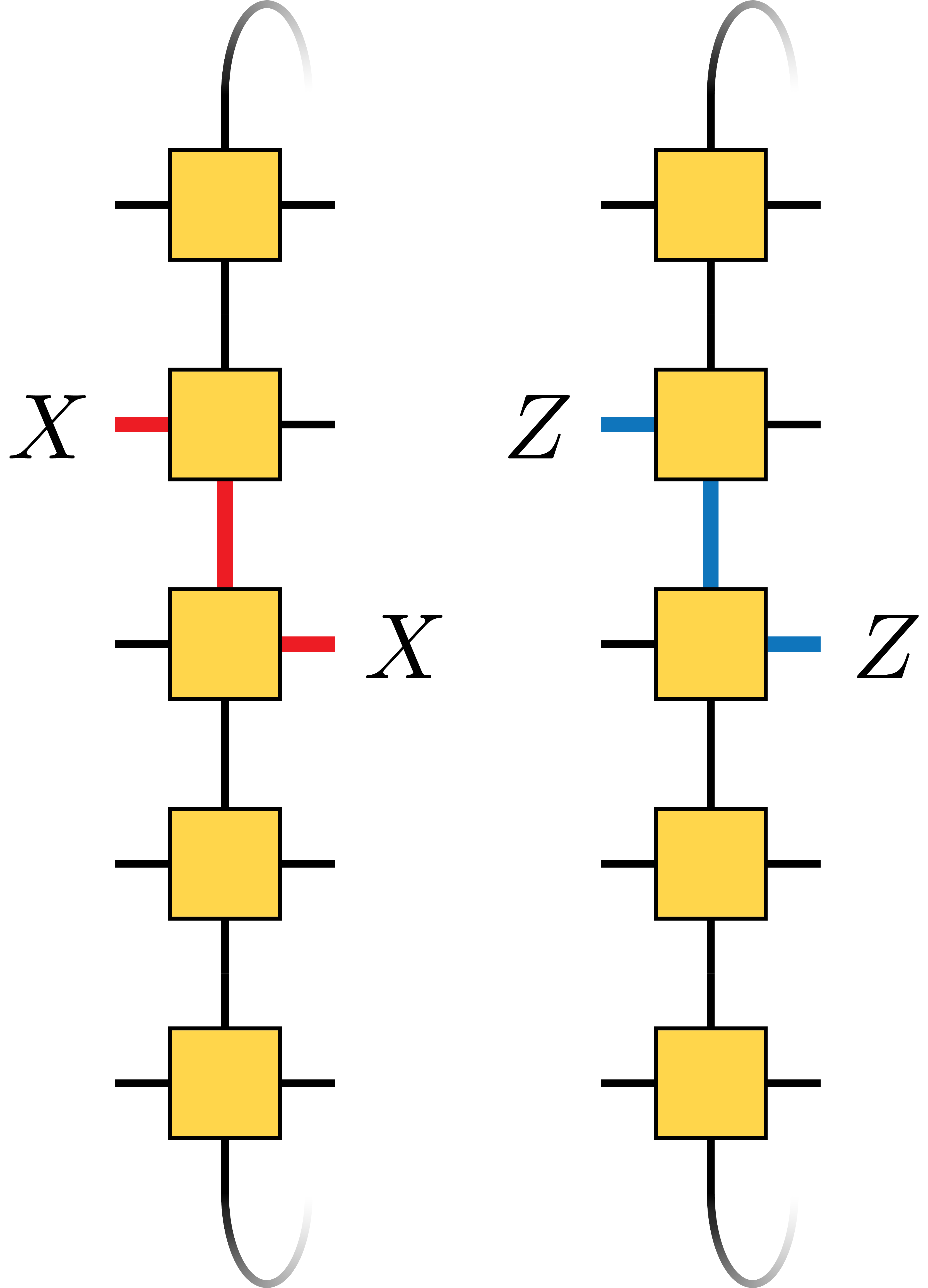} \\
        \rule[-5mm]{0pt}{9mm} (iii) & (iv)
    \end{tabular}
    \caption{Examples of each class of local generators from Lemma \ref{lem:HV summary}.}
    \label{fig:islands}
\end{figure}

\subsubsection{Bipartite local generators}

So far we have presented two forms of $S(A,n,1)$, each with advantages and disadvantages. The canonical bipartite form directly reveals the channel capacity of $\mathcal{A}(A,n,1)$ and determines the channel capacity for any depth via $\Phi_1$ and $\Omega_1$. However, in general, it may depend on $n$. The local form removes this dependence on $n$ but does not tell us anything about channel capacity. In this section, we connect the local form and the canonical form.

The first step is to partition the local generators into the same left and right subsets used for the canonical form. When viewed on one partition at a time, the local generators are no longer guaranteed to be abelian; the degree to which they fail to be abelian dictates $\Phi_1$. Without loss of generality, consider the left partition. A computer search reveals that up to gauge equivalence there are nine local generating sets of $S_L$ (Table \ref{tab:islands}), which we refer to as the standard local generators. 

Recall that equivalence classes of $\Phi_1$ are labeled by the tuple $(n,p)$, where $p$ is the number of anticommuting pairs. It follows from \eqref{eq:canonical stabs} that $p = \rank(S_L) - n.$ The rank of $S_L$ is independent of any particular generating set, so we may use the local generators to determine $p$ as a function of $n$. When this is done, we find a one-to-one correspondence between the non-trivial equivalence classes of $\Phi_1$ -- that is, classes where $p\neq 0$ -- and the standard local generating sets (Table \ref{tab:islands}). Furthermore, by comparing $p(n)$ with $C(A,n,1)$ in Fig. \ref{fig:heatmaps} we can see that each transmission capacity class corresponds to exactly one equivalence class of $\Phi_1$ (Table \ref{tab:phi to tp}).

\begin{table}
 \centering
 \footnotesize
 \setlength\tabcolsep{1.5pt} 
 \begin{tabular}{|c|c|l|l|}
 \hline
 \thead{$\Phi_1$ \\ Class} & \thead{Standard \\ Local \\ Generators} & \thead{Rank of $S_L$} & \thead{$p$} \\ \hline
 \textbf{a} & $\langle Z_i \rangle$ & $n$ & 0 \\ \hline
 \textbf{a} & $\langle X_{i-1}Z_iX_{i+1} \rangle$ & $n$ & 0 \\ \hline
 \textbf{a} & $\langle Z_i Z_{i+1}, \prod_i X_i \rangle$ & $n$ & 0 \\ \hline
 \textbf{b} & $\langle Z_i, \prod_i X_i \rangle $ & $n+1$ & 1 \\ \hline
 \textbf{c} & $\langle Z_i, X_iX_{i+1}\rangle$ & $2n-1$ & $n-1$ \\ \hline
 \textbf{d} & $\langle Z_i, X_{i-1}I_i X_{i+1} \rangle$ & $\begin{cases} 2n-2, \text{ $n$ even} \\ 2n-1, \text{ $n$ odd} \end{cases}$ & $\begin{cases} n-2 \\ n-1 \end{cases}$\\ \hline
 \textbf{e} & $\langle Z_i, X_{i-1}X_iX_{i+1} \rangle$ & $\begin{cases} 2n-2, \text{ $n$ div. by 3} \\ 2n, \text{ otherwise} \end{cases}$ & $\begin{cases} n-2 \\ n \end{cases}$\\ \hline
 \textbf{f} & $\langle Z_i, X_iX_{i+1}, \prod_i X_i\rangle$ & $\begin{cases} 2n-1 \quad n \text{ even} \\ 2n \quad n \text{ odd} \end{cases}$ & $\begin{cases} n-1 \\ n \end{cases}$\\ \hline
 \textbf{g} & $\langle Z_i, X_i\rangle$ & $2n$ & $n$ \\ \hline
 \end{tabular}
 \caption{Standard local generating sets of $S_L$.}
 \label{tab:islands}
\end{table}

\begin{table}
\centering
  \begin{tabular}{|c|l|}
    \hline
    \thead{$\Phi_1$ \\ Class} & \thead{Transmission \\ Classes} \\
    \hline
    \textbf{a} & 0 \\ \hline
    \textbf{b} & 1, 2 \\ \hline 
    \textbf{c} & 4, 6, 7 \\ \hline
    \textbf{d} & 3, 5 \\ \hline
    \textbf{e} & 8, 10 \\ \hline
    \textbf{f} & 9, 11 \\ \hline
    \textbf{g} & 12 \\ \hline
  \end{tabular}
  \caption{Correspondence between transmission classes and equivalence classes of $\Phi_1$.}
  \label{tab:phi to tp}
\end{table}

\begin{definition}
 Let the vector $\va{\Phi}(A)$ be 
 \begin{equation}
 \va{\Phi}(A) = \Phi_1(A, n) \quad n = 1, 2, \dots, \infty
 \end{equation}
 and define $\va{\Phi}(A_1)\sim\va{\Phi}(A_2)$ if $\Phi_1(A_1, n) \sim \Phi_1(A_2, n)$ for all $n$.
\end{definition}
\begin{lemma} \label{lem:phi vector}
 $\va{\Phi}(A_1)\sim\va{\Phi}(A_2)$ if and only if $p(A_1,n) = p(A_2,n)$ for values of $n$ such that at least one is even, one is odd, one is divisible by 3, and one is not divisible by 3.
\end{lemma}
\begin{proof}
 If $\va{\Phi}(A_1) \sim \va{\Phi}(A_2)$ then by definition $p(A_1,n)=p(A_2,n)$ for all $n$. To prove the converse, notice that for every standard local generatoring set in Table \ref{tab:islands}, the expression for $p$ as a function of $n$ only changes if $n$ switches from even to odd, or from being divisible by 3 to not divisible by 3. Furthermore, these dependencies are mutually exclusive, i.e., if $p$ depends on whether $n$ is even or odd then it does not depend on the divisibility of $n$ by 3, and vice versa. Therefore, if $p(A_1,n) = p(A_2,n)$ for values of $n$ that satisfy the conditions of the lemma, then $p(A_1,n)=p(A_2,n)$ for all $n$, which implies $\va{\Phi}(A_1) \sim \va{\Phi}(A_2)$.
\end{proof}

\subsubsection{Periodic boundary conditions}

The local generators of $S_L$ and $S_R$ act on a ring of $n$ sites, thus they are subject to the effects of periodic boundary conditions. These effects have already played a role in the characterization of $\va{\Phi}(A)$, e.g. in Table \ref{tab:islands} and Lemma \ref{lem:phi vector}. The next object we wish to characterize is $\Omega_1(A,n)$, but to do so we require some discussion of local commutation relations in the presence of PBC. 

Consider two local operators $B$ and $C$ acting on a ring of $n$ sites. The intersection between $B$ and $C$ is the overlap between the support of $B$ and the support of $C$ for every translation of $C$. Following \cite{Schumacher2004}, we define the intersection between $B$ and $C$ to be regular for a given value of $n$ if its geometry does not change as $n\to\infty$.

For example, let $B_i = X_i X_{i+1} X_{i+2}$ and $C_i = Z_i Z_{i+1} Z_{i+2}$. When $n$ is large $B_i$ and $C_j$ slide past each other with no boundary effects. However, when $n=4$ both ends of $C_3$ simultaneously overlap with $B_1$. This type of intersection is not possible when $n\to\infty$, so the intersection of $B$ and $C$ is irregular for $N=4$. 

Regularity of intersection has a direct effect on commutation relations, as illustrated by the fact that $B_1$ and $C_3$ commute for $N=4$ and anticommute for $n\geq 5$. To formalize this phenomenon, we have the following lemma.

\begin{lemma} \label{lem:regularity}
 Consider a $k$-local operator $B_i$ and an $\ell$-local operator $C_j$ acting on a ring of $n$ sites. The intersection between $B_i$ and $C_j$ is irregular for $n < k + \ell - 1$ and regular otherwise.
\end{lemma}
\begin{proof}
 Let $B_i$ have support on sites $i$ through $i+k-1$ and $C_j$ have support on sites $j$ through $j+\ell-1$. The transition between irregular and regular intersection is given by the case where $B_i$ and $C_k$ only intersect at their endpoints, i.e. at sites $k$ and $i$. By inspection, this requires $n = k + \ell - 2$, therefore the intersection is irregular for $n < k + \ell - 1$.
\end{proof}
\begin{corollary} \label{cor:regular commutation}
 The commutation relations between $B_i$ and $C_j$ change if and only if $n$ transitions from the irregular to the regular regime, or vice versa. 
\end{corollary}
\begin{proof}
 The commutation relations are determined by the overlap between $B_i$ and $C_j$, and this overlap only changes if there is a transition from irregular to regular intersection, or vice versa. 
\end{proof}
\begin{corollary} \label{cor:geq 5}
 The commutation relations between the local generators of $S_L(A,n,d)$ and $S_R(A,n,d)$ are regular for $n \geq 5$, and irregular for $n<5$. 
\end{corollary}
\begin{proof}
 Apart from the string-like operators $\prod_i \bar{g}_i$, every local generator has support on at most three adjacent sites. Applying the result of Lemma \ref{lem:regularity}, we have regular intersection for $n \geq 3 + 3 - 1 = 5$. The string-like operators have periodicity one, so their intersection with the rest of the island stabilizers is regular for $n \geq 3 + 1 - 1 = 3$. Therefore, the special case of string-like operators does not violate the general result.
\end{proof}

\subsubsection{Removing the dependence on $n$}

\begin{table}[ht]
  \centering
  \footnotesize
  \setlength\tabcolsep{1.5pt} 
  \begin{tabular}{|c|c|}
  \hline
  \thead{Standard \\ Local \\ Generators} & \thead{Generators of $\mathcal{Z}_L$} \\ \hline
  $\langle Z_i \rangle$ & $\langle Z_i \rangle$ \\ \hline
  $\langle X_{i-1}Z_iX_{i+1} \rangle$ & $\langle X_{i-1}Z_iX_{i+1} \rangle$ \\ \hline
  $\langle Z_i Z_{i+1}, \prod_i X_i \rangle$ & $\langle Z_i Z_{i+1}, \prod_i X_i \rangle$ \\ \hline
  $\langle Z_i, \prod_i X_i \rangle $ & $\langle Z_i Z_{i+1} \rangle $ \\ \hline
  $\langle Z_i, X_iX_{i+1}\rangle$ & $\langle \prod_i Z_i\rangle$ \\ \hline
  $\langle Z_i, X_{i-1}I_i X_{i+1} \rangle$ & $\begin{cases} \langle \prod_i Z_{2i}, \prod_i Z_{2i+1} \rangle &\text{$n$ even} \\ \langle \prod_i Z_i\rangle &\text{$n$ odd} \end{cases}$ \\ \hline
  $\langle Z_i, X_{i-1}X_iX_{i+1} \rangle$ & $\begin{cases} \langle \prod_i Z_{3i}Z_{3i+1}, \prod_i Z_{3i+1}Z_{3i+2} \rangle, \\ \qquad\qquad\qquad\qquad\text{$n$ div. by 3} \\ \langle I\rangle, \text{ otherwise} \end{cases}$ \\ \hline
  $\langle Z_i, X_iX_{i+1}, \prod_i X_i\rangle$ & $\begin{cases} \langle \prod_i Z_i \rangle &\text{$n$ even} \\ \langle I\rangle &\text{$n$ odd} \end{cases}$ \\ \hline
  $\langle Z_i, X_i\rangle$ & $\langle I\rangle$ \\ \hline
  \end{tabular}
  \caption{Generators of $\mathcal{Z}_L$ for each standard local generating set of $S_L$.}
  \label{tab:ZL}
\end{table}

\begin{definition}
  Let the vector $\va{\Omega}(A)$ be
  \begin{equation}
  \va{\Omega}(A) = \Omega_1(A,n) \quad n=1,\dots,\infty,
  \end{equation}
  and define $\va{\Omega}(A_1) \sim \va{\Omega}(A_2)$ if $\Omega_1(A_1,n) \sim \Omega_1(A_2,n)$ for all $n$.
\end{definition}
\begin{lemma} \label{lem:omega vector}
  $\va{\Omega}(A_1) \sim \va{\Omega}(A_2)$ if and only if $\Omega_1(A_1,n) \sim \Omega_1(A_2,n)$ for $n\leq 6$. 
\end{lemma}
\begin{proof}
  By definition, $\va{\Omega}(A_1) \sim \va{\Omega}(A_2)$ implies $\Omega_1(A_1,n) \sim \Omega_1(A_2,n)$ for all $n$. To prove the converse, recall that $\Omega_1(A,n)$ is the set of commutators of the elements of $\mathcal{Z}_L(A,n,1)$ with the elements of $\mathcal{Z}_R(A,n,1)$ and $P_R(A,n,1)$. Furthermore, $\mathcal{Z}_L$ dictates the elements of $\mathcal{Z}_R$ and $P_R$. To see this, note that $P_L$ completes the centralizer of $\mathcal{Z}_L$, thus $\mathcal{Z}_L$ uniquely determines $P_L$. Then, $P_L$ uniquely determines $P_R$ by taking the right-hand part of each generator of $P_L$. Finally, $P_R$ uniquely determines $\mathcal{Z}_R$, because $\mathcal{Z}_R$ is the center of $P_R$. 
 
  In Table \ref{tab:ZL}, we have an expression for the generators of $\mathcal{Z}_L$ for every equivalence class of local tensors. These generating sets show that the elements of $\mathcal{Z}_L(A,n)$ and $\mathcal{Z}_L(A,m)$ are equal up to translation, as long as $n$ and $m$ share the same divisibility by 2 or 3. Therefore, the elements of every group relevant to $\Omega_1(A,n)$ remain constant up to translation, as long as $n$ retains the same divisibility by 2 or 3.
 
  Having characterized the elements of $\mathcal{Z}_L$, $\mathcal{Z}_R$ and $P_R$ up to translation, it remains to deal with their commutation relations. This is not trivial because periodic boundary conditions may play a role. However, every element of $\mathcal{Z}_L$, $\mathcal{Z}_R$, and $P_R$ may be written as a product of local generators. Thus, the commutator $\comm{x}{y}$ for $x \in \mathcal{Z}_L$ and $y\in \mathcal{Z}_R\cup P_R$ has a decomposition into pairwise commutators of local generators. By Corollaries \ref{cor:regular commutation} and \ref{cor:geq 5}, these pairwise commutators may vary in the irregular regime of intersection ($n \leq 4$), but they are invariant in the regular regime of intersection ($n\geq 5$). In other words, for any $n\geq 5$, every local generator anticommutes with the same finite subset of its neighbors.
 
  It follows that if $n_0 \geq 5$ and $\Omega_1(A_1,n_0) \sim \Omega_1(A_2,n_0)$, then $\Omega_1(A_1,n) \sim \Omega_1(A_2,n)$ for all $n$ which have the same divisibility by 2 and 3 as $n_0$. Therefore, to guarantee $\Omega_1(A_1,n) \sim \Omega_1(A_2,n)$ for all $n$, it suffices to check just a few values of $n_0$ and infer the equivalence of the $\Omega$s for all other $n$. Coverage of every $n$ in the regular regime can be obtained by checking at least one value of $n$ even, one value of $n$ odd, one divisible by 3, and one not divisible by 3. The values $n=5,6$ satisfy these requirements. 
  
  We have made no statement about the irregular regime ($n\leq 4$), thus we must also check every $n$ in the irregular regime. Putting everything together, we have that if $\Omega_1(A_1,n) \sim \Omega_1(A_2,n)$ for $n\leq 6$, then they are equivalent for all $n$. 
\end{proof}

The standard form of $\Omega_1(A,n)$ was computed for $n\leq 6$ and every stabilizer tensor, and it was found that there are 19 different cases. By Lemma \ref{lem:omega vector}, this implies that there are 19 equivalence classes of $\va{\Omega}$. With this result, we have a classification of both $\va{\Phi}$ and $\va{\Omega}$. The last piece of the puzzle is to show that if two stabilizer tensors share the same $\va{\Phi}$ and $\va{\Omega}$, then they are guaranteed to share the same channel capacity for any values of $n$ and $d$. 

Upon comparing the $\va{\Omega}$ class and the transmission class of each stabilizer tensor, it was found that there are 7 different $\va{\Omega}$'s that correspond to transmission Class 0 and one unique $\va{\Omega}$ for each of the other transmission classes. Thus, we can assign a unique tuple $(\va{\Phi},\va{\Omega})$ to the stabilizer tensors of each non-trivial transmission class. By Algorithm 1, $\va{\Phi}$ and $\va{\Omega}$ are sufficient to compute $C(A,n,d)$ for any $n$ and $d$, therefore Proposition \ref{prop:wire} must be true for every non-trivial transmission class.

Finally, we deal with the trivial transmission class, i.e. $C(A,n,d)=0$. By Lemma \ref{lem:step_down}, $C(A,n,1)$ cannot increase with depth, and Tables \ref{tab:islands} and \ref{tab:phi to tp} show that every $\va{\Phi}$ associated with transmission Class 0 gives $C(A,n,1)=0$ for all $n$. Therefore, if $A$ is in transmission Class 0, then $C(A,n,d)=0$ for all $n$ and $d$. This completes the proof of Proposition \ref{prop:wire}.

\section{Conclusion} \label{sec:concl}

We have reviewed three interesting computational tasks that may be accomplished via local measurements on stabilizer PEPS: quantum wire, quantum computation on a resource state, and quantum computation in an SPT phase. Classifying the capability of stabilizer PEPS with respect to these tasks would have two beneficial outcomes. First, it would help to place well-known states like the cluster state into perspective: are they exceptionally useful for quantum computation, or are they just one species in a zoo of useful states? Second, by revealing the differences between useful and not-so-useful classes of states, it might offer insights into the nature of quantum computation.

In this paper, we have taken the first step in a three-part classification program by presenting a complete classification of quantum wire in cylindrical stabilizer PEPS on a 2D square lattice. In particular, we show that the transmission capacity fits into 13 distinct classes, including the class of Clifford QCA where $C(n,d)=n$, the trivial class where $C(n,d)=0$, and 11 intermediate classes. 

It would also be interesting to explore the transmission capacity of stabilizer PEPS with different lattice geometries and boundary conditions, but we leave this direction for future work. With our original geometry, the next step is to classify the computational power of stabilizer PEPS as resource states for measurement based quantum computation. To be more precise: given a stabilizer PEPS, what is the set of unitary gates that may be effected by local measurement? Finally, one might ask if it is possible to construct an SPT phase around every stabilizer PEPS. If so, is the computational power uniform throughout those phases?

\section*{Acknowledgements}

This research was undertaken thanks, in part, to funding from the Canada First Research Excellence Fund, Quantum Materials and Future Technologies Program. PH is funded by the National Science and Engineering Research Council of Canada (NSERC). RR is supported through NSERC and US ARO (W911NF2010013).

\bibliographystyle{quantum}
\bibliography{2dWire}

\onecolumn\newpage
\appendix

\section{Proof of Lemma \ref{lem:HV p=1 q=1}} \label{apx:p1q1}

\mylemma*

\begin{proof}[Proof of (i)]
  Recall from \eqref{eq:RA} that when $q=1$, $$R_A = \expval{\gamma^u \otimes \gamma^d, \bar{\gamma}^u \otimes \bar{\gamma}^d}.$$ If $\gamma^u=\gamma^d=\gamma$, we have $v = \gamma\otimes\gamma$ and $g^V = \bar{\gamma}^u\otimes\bar{\gamma}^d$, where $\gamma$, $\bar{\gamma}^u$, and $\bar{\gamma}^d$ are different Pauli operators. Let $\gamma = X$, $\bar{\gamma}^u = Z$, and $\bar{\gamma}^d=Y$ without loss of generality. This gives $v = X\otimes X$ and $g^V = Z\otimes Y$. The only 2-qubit Pauli operators that commute with $X\otimes X$ but anticommute with $Z\otimes Y$ are $X\otimes I$ and $I\otimes X$, which are in the same coset with respect to $\{e,v\}$. Therefore, $\bar{g}^V = X\otimes I$ completes the group $S_A^{(V)}$. The cosets with respect to $\{e,v\}$ are:
  \begin{equation}
  \begin{aligned}
  &\qty[v]_v = \{X\otimes X, I\otimes I\} \\
  &\qty[g^V]_v = \{Z \otimes Y, Y \otimes Z\} \\
  &\qty[\bar{g}^V]_v = \{X\otimes I, I\otimes X\} \\
  &\qty[g^V \bar{g}^V]_v = \{Y\otimes Y, Z\otimes Z\}.
  \end{aligned}
  \end{equation}
 
  It will be convenient to re-label the cosets cyclicly, i.e. $g^V \to \bar{g}^V \to g^V\bar{g}^V \to g^V$, because this gives the vertical tilings $\prod_i \bar{g}^V_i$ and $g^V_i g^V_{i+1}$. Reading off the action of these tilings on the horizontal partition from \eqref{eq:HV canonical form}, we have $\prod_i \bar{g}^H_i$ and $g^H_i g^H_{i+1}$ as elements of $S(A,n,1)$. The former accounts for one generator, and the latter for $n-1$ independent generators. Finally, $h_i$ provide the remaining $n$ generators.
 \end{proof}
 \begin{proof}[Proof of (ii)]
  If $\gamma^u \neq \gamma^d$, then $v = \gamma^u \otimes \gamma^d$. Without loss of generality, let $\gamma^u = X$ and $\gamma^d = Z$ so that $v = X\otimes Z$. Once again we must find three independent cosets that commute with $v$ but anticommute pairwise:
  \begin{equation} \label{eq:HV p=1 q=1 neq}
  \begin{aligned}
  &\qty[v]_v = \{X\otimes Z, I \otimes I\} \\
  &\qty[g^V]_v = \{Y \otimes Y, Z \otimes X\} \\
  &\qty[\bar{g}^V]_v = \{X\otimes I, I\otimes Z\} \\
  &\qty[g^V \bar{g}^V]_v = \{Z\otimes Y, Y\otimes X\}.
  \end{aligned}
  \end{equation}
  Upon inspection of \eqref{eq:HV p=1 q=1 neq}, the vertical tiling $\bar{g}^V_{i-1}g^V_i\bar{g}^V_{i+1}$ is a local operator. The action of this operator on the horizontal partition gives $\bar{g}^H_{i-1}g^H_i\bar{g}^H_{i+1} \in S(A,n,1)$. This is a set of $n$ independent and commuting generators, and $h_i$ provide the remaining $n$ generators.
 \end{proof}

\end{document}